%% file: main.tex
\PassOptionsToPackage{warn}{textcomp}
\pdfoutput=1
\documentclass[sigconf,9pt]{acmart}
\newcommand{\fullVersion}{}
\newcommand{\arXivVersion}{}

\usepackage[english]{babel}
\usepackage{blindtext}
\usepackage{ulem}
\usepackage{bmpsize}

\newcommand{\xVal}{\ensuremath{0.5\cdot \parentheses{Z-1 + \ln(k/\delta) + \sqrt{(Z-1 + \ln(k/\delta))^2-(Z-1)^2/4}}}}
\newcommand{\NVal}{\ensuremath{k\cdot \parentheses{Z-1 + \ln(k/\delta) + \sqrt{(Z-1 + \ln(k/\delta))^2-(Z-1)^2/4}}}}

\iftrue
\newcommand{\paraspace}{\vspace{0.00in}}
\usepackage{amssymb}
\usepackage{pifont}
\newcommand{\cmark}{\ding{51}}%
\newcommand{\xmark}{\ding{55}}%

\fi

\usepackage[font=small,labelfont=bf]{caption}
\usepackage{times,xcolor,xspace}

\usepackage{algorithm}
\usepackage{subfigure}

\newcommand\duparrow{\mathbin{\uparrow\uparrow}}

\usepackage[capitalise]{cleveref}
\usepackage{amsmath}
\usepackage{amsthm}
\usepackage[noend]{algpseudocode}

\newcommand{\algSize}{\normalfont}

\usepackage{multirow}

\input{macros}
\usepackage{enumitem}
\usepackage{comment}

\ifdefined\arXivVersion
\else
\usepackage{hyperref}
\hypersetup{pdfstartview=FitH,pdfpagelayout=SinglePage}
\fi

\ifdefined\noComments
\newcommand{\minlan}[1]{}
\newcommand{\mm}[1]{}
\newcommand{\Michael}[1]{}
\newcommand{\ran}[1]{}
\newcommand{\gianni}[1]{}
\else
\newcommand{\minlan}[1]{\textcolor{blue}{(Minlan: #1)}}
\newcommand{\mm}[1]{\textcolor{red}{(MM: #1)}}
\newcommand{\Michael}[1]{\textcolor{red}{(MM: #1)}}
\newcommand{\ran}[1]{\textcolor{purple}{(Ran: #1)}}
\newcommand{\gianni}[1]{\textcolor{cyan}{(Gianni: #1)}}

\fi
\newcommand{\sys}{\ensuremath{\mathit{PINT}}\xspace}

\newcommand{\parab}[1]{\paraspace\noindent{\bf #1} }

\ifdefined\arXivVersion
\renewcommand\footnotetextcopyrightpermission[1]{} 
\else
\copyrightyear{2020}
\acmYear{2020}
\setcopyright{acmcopyright}\acmConference[SIGCOMM '20]{Annual conference of the ACM Special Interest Group on Data Communication on the applications, technologies, architectures, and protocols for computer communication}{August 10--14, 2020}{Virtual Event, NY, USA}
\acmBooktitle{Annual conference of the ACM Special Interest Group on Data Communication on the applications, technologies, architectures, and protocols for computer communication (SIGCOMM '20), August 10--14, 2020, Virtual Event, NY, USA}
\acmPrice{15.00}
\acmDOI{10.1145/3387514.3405894}
\acmISBN{978-1-4503-7955-7/20/08}
\fi
\setcopyright{acmcopyright}

\acmPrice{}

\iffalse 
\newcommand{\changed}[1]{\textcolor{blue}{#1}}
\else
\newcommand{\changed}[1]{{#1}}
\renewcommand{\sout}[1]{}
\renewcommand{\textcolor}[2]{}
\fi


\begin{document}
\title{\sys: Probabilistic In-band Network Telemetry}


\normalem
\author{Ran Ben Basat}
\affiliation{%
  \institution{Harvard University}
}
\email{ran@seas.harvard.edu}
\author{Sivaramakrishnan Ramanathan}
\affiliation{%
  \institution{University of Southern California}
}
\email{satyaman@usc.edu}
\author{Yuliang Li}
\affiliation{%
  \institution{Harvard University}
}
\email{yuliangli@g.harvard.edu}
\author{Gianni Antichi}
\affiliation{%
  \institution{Queen Mary University of London}
}
\email{g.antichi@qmul.ac.uk}
\author{Minlan Yu}
\affiliation{%
  \institution{Harvard University}
}
\email{minlanyu@seas.harvard.edu}
\author{Michael Mitzenmacher}
\affiliation{%
  \institution{Harvard University}
}
\email{michaelm@eecs.harvard.edu}


\renewcommand{\shortauthors}{Ben Basat et al.}

\input{abstract}

\maketitle
\ifdefined\arXivVersion
\pagestyle{plain}
\fi

\input{introduction}

\input{motivation}

    \bibliographystyle{ACM-Reference-Format}
    \clearpage
    \bibliography{references}
    
    \ifdefined\fullVersion
    \clearpage
    \appendix
    \input{Appendix.tex}

    \fi
\end{document}

%% file: macros.tex
\renewcommand{\epsilon}{\varepsilon}

\newcommand{\set}[1]{\left\{#1\right\}}
\newcommand{\ceil}[1]{\left\lceil#1\right\rceil}

\newtheorem{theorem}{Theorem}
\newtheorem{fact}[theorem]{Fact}

\newcommand{\floor}[1]{\left\lfloor#1\right\rfloor}
\newcommand{\brackets}[1]{\left[#1\right]}

\newcommand{\parentheses}[1]{\left(#1\right)}
\algnewcommand\algorithmicforeach{\textbf{for each}}
\algdef{S}[FOR]{ForEach}[1]{\algorithmicforeach\ #1\ \algorithmicdo}

%% file: abstract.tex
\begin{abstract}
Commodity network devices support adding in-band telemetry measurements into data packets, enabling a wide range of applications, including network troubleshooting, congestion control, and path tracing. However, including such information on packets adds significant overhead that impacts both flow completion times and {application-level performance.}

We introduce \sys, an in-band telemetry framework that bounds the amount of information added to each packet. 
\sys encodes the requested data on multiple packets, allowing per-packet overhead limits that can be as low as one bit. 
We analyze \sys and prove performance bounds, 
including cases when multiple queries are running simultaneously. 
\sys is implemented in P4 and can be deployed on network devices.
Using real topologies and traffic characteristics,
\changed{we show that \sys \sout{can enable the above applications}
concurrently enables applications such as congestion control, path tracing, and computing tail latencies,} using only sixteen bits per packet, with {performance comparable \mbox{to the state of the art.}}

%
%
\end{abstract}

%% file: introduction.tex
\section{Introduction}
Network telemetry is the basis for a variety of network management applications
such as network health monitoring~\cite{tammana16}, debugging~\cite{guo15}, fault 
localization~\cite{arzani2018007}, resource accounting and planning~\cite{narayana16}, 
attack detection~\cite{savage00,gkounis16}, congestion control~\cite{li19}, load 
balancing~\cite{alizadeh14,katta16,katta17}, fast reroute~\cite{liu13}, and path 
tracing~\cite{jeyakumar14}. 
A significant recent advance is provided by the In-band Network Telemetry (INT)~\cite{int}. 
INT allows switches to add information to each packet,
such as switch ID, link utilization, or queue status, as it passes by. Such telemetry information is then collected  
at the network egress point upon the reception of the packet.  

INT is readily available in programmable switches and network interface cards (NICs)~\cite{broadcom-int,barefoot-int,netronome-int,xilinx-int}, enabling an unprecedented level of visibility into the data plane behavior and making this technology attractive for real-world deployments~\cite{att-int,li19}.
A key drawback of INT is the overhead on packets. Since each switch adds information to the packet, the packet byte overhead grows linearly with the path length.
Moreover, the more telemetry data needed per-switch, the higher the overhead is: on a generic data center topology with 5 hops, requesting two values per switch requires 48 Bytes of overhead, or 4.8\% of a 1000 bytes packet (\S\ref{sec:motiv}).
When more bits used to store telemetry data, fewer bits can be used to carry the packet payload and stay within the maximum transmission unit (MTU). 
As a result, applications may have to split a message, e.g., an RPC call, onto multiple packets, making it harder to support the run-to-completion model that high-performance transport and NICs need~\cite{barbette15}. 
\changed{Indeed, the overhead of INT can impact application performance, potentially leading in some cases to a 25\% increase and 20\% degradation of flow completion time and goodput, respectively (\S\ref{sec:motiv}).
Furthermore, it increases processing latency at switches and might impose additional challenges for collecting and processing the data (\S\ref{sec:motiv}).}

We would like the benefits of in-band network telemetry, 
but at smaller overhead cost;  in particular, we wish to 
minimize the per-packet bit overhead.  
We design Probabilistic In-band Network Telemetry (\sys), a probabilistic variation of INT, that provides similar visibility as INT 
while bounding the per-packet overhead according to limits set
by the user. 
\changed{\sys allows the overhead budget to be as low as one bit, and leverages approximation techniques to meet it. We argue that often an approximation of the telemetry data suffices for the consuming application. For example, telemetry-based congestion control schemes like HPCC~\cite{li19} can be tuned to work with approximate telemetry, as we demonstrate in this paper.}
In some use cases, a single bit per packet suffices.

With \sys, a query is associated with a maximum overhead allowed on each packet. The requested information is probabilistically encoded onto several different packets so that a \emph{collection} of a flow's packets provides the relevant data. 
In a nutshell, while with INT a query triggers every switch along the path to embed their own information, \sys spreads out the information over multiple packets to minimize the per-packet overhead. The insight behind this approach is that, 
for most applications, it is not required to know all of the per-packet-per-hop information that INT collects. 
\changed{existing techniques incur high overheads due to requiring perfect telemetry information. For applications where some imperfection would be sufficient, these techniques may incur unnecessary overheads. \sys Is designed for precisely such applications}
For example, it is possible to check a flow's path conformance~\cite{handigol14,tammana16,narayana16}, by inferring its path from a collection of its packets. Alternatively, congestion control or load balancing algorithms that rely on latency measurements gathered by INT, e.g., HPCC~\cite{li19}, Clove~\cite{katta17} can work if packets convey information about the path's bottleneck, and do not \mbox{require information about all hops.}

We present the \sys framework (\S\ref{sec:framework}) and show that it can run several concurrent queries while bounding the per-packet bit overhead. 
To that end, \sys uses each packet for a query subset with cumulative overhead within the user-specified budget. We introduce the techniques we used to build this solution (\S\ref{sec:algorithms}) alongside its implementation on commercial programmable switches supporting P4 (\S\ref{sec:impl}).
Finally, we evaluate (\S\ref{sec:eval}) our approach with three different use cases.
The first traces a flow's path, the second uses data plane telemetry for congestion control, 
and the third estimates the experienced median/tail latency. 
Using real topologies and 
traffic characteristics, we show that \sys enables all of them concurrently, with
only sixteen bits per packet and while providing comparable performance to the state of the art.
%

%

In summary, the main contributions of this paper are:
\begin{itemize}[leftmargin=*]
\item We present \sys, a novel in-band network telemetry approach that provides fine-grained visibility while bounding the per-packet \mbox{bit overhead to a user-defined value.}
\item We analyze \sys and rigorously prove performance bounds.
\item We evaluate \sys in on path tracing, congestion control, and latency estimation, over multiple \mbox{network topologies.}
\item \changed{We open source our code~\cite{code}}.
\end{itemize}

%% file: motivation.tex
\begin{table}
{\small
        \begin{tabular}{l|l}
                \hline
                \hspace{7mm}\textbf{Metadata value} & \hspace{12mm}\textbf{Description} \\
                \hline
                \hline
                Switch ID                       & ID associated with the switch\\
                Ingress Port ID                 & Packet input port\\
                Ingress Timestamp               & Time when packet is received\\
                Egress Port ID                  & Packet output port                \\
                Hop Latency                     & Time spent within the device      \\
                Egress Port TX utilization      & Current utilization of output port\\
                Queue Occupancy                 & The observed queue build up       \\
                Queue Congestion Status         & Percentage of queue being used    \\
                \hline
        \end{tabular}
        }
        \caption{Example metadata values.}
        \label{tab:int-info}
        
\end{table}

\section{INT and its Packet Overhead}
\label{sec:motiv}
INT is a framework designed to allow the collection and reporting of network data plane status at switches, without requiring any control plane intervention. In its architectural model, designated INT traffic sources, (e.g., the end-host networking stack, hypervisors, NICs, or ingress switches), add an \emph{INT metadata header} to packets.
The header encodes
\textit{telemetry instructions} that are followed by network devices on the packet's path. These instructions tell an INT-capable device what information to add to packets as they transit the network. 
Table~\ref{tab:int-info} summarizes the supported \emph{metadata values}. 
Finally, INT traffic sinks, e.g., egress switches or receiver hosts, retrieve 
the collected results before delivering the original packet to the application.
\changed{The INT architectural model is intentionally generic, and hence can 
enable a number of high level applications, such as (1) Network troubleshooting 
and verification, i.e., microburst detection~\cite{jeyakumar14}, packet history~\cite{handigol14}, 
path tracing~\cite{jeyakumar14}, path latency computation~\cite{IvkinYB019}; (2) 
Rate-based congestion control, i.e., RCP~\cite{dukkipati06}, XCP~\cite{katabi02}, 
TIMELY~\cite{mittal15}; (3) Advanced routing, i.e, utilization-aware load balancing~\cite{alizadeh14,katta16}.}

%
\changed{INT imposes a non insignificant overhead on packets though.}
The metadata header is defined as an 8B vector specifying the telemetry requests. 
Each value is encoded with a 4B number, as defined by the protocol~\cite{int}.
As INT encodes per-hop information, the overall overhead grows linearly with both the number of metadata values and the number of hops. For a generic data center topology with 5 hops, the minimum space required on packet would be 28 bytes (only one metadata value per INT device), which is 2.8\% of a 1000 byte packet (e.g., RDMA has a 1000B MTU). 
Some applications, such as Alibaba's High Precision Congestion Control~\cite{li19} (HPCC), require three different 
INT telemetry values for each hop. 
Specifically, for HPCC, INT collects timestamp, egress port tx utilization, and queue occupancy, alongside some additional data that is not defined by the INT protocol. 
This would account for around 6.8\% overhead using a standard INT \mbox{on a 5-hop path.}\footnote{HPCC reports a slightly lower (4.2\%) overhead
because they use customized INT. For example, they do not use the INT
header as the telemetry instructions do \mbox{not change over time.}}
\changed{\mbox{This overhead poses several problems:}}

\begin{figure}[tb]\centering
        \begin{minipage}{0.49\columnwidth}\centering
                \includegraphics[width=1\columnwidth]{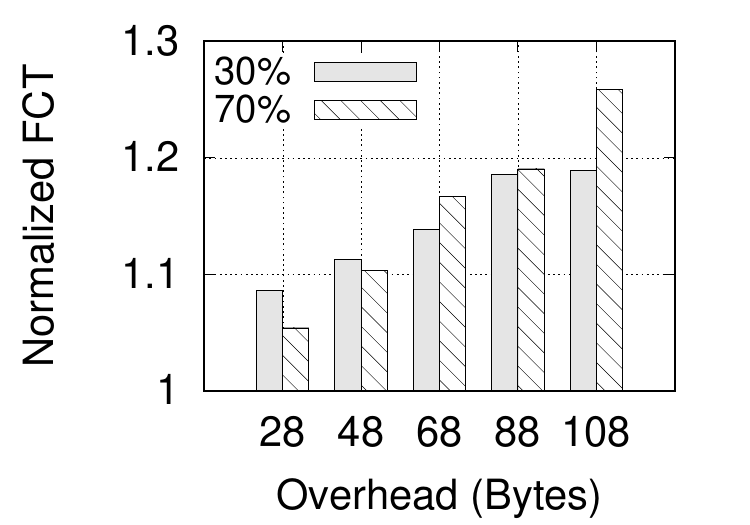}
        \end{minipage}\hfill\begin{minipage}{0.49\columnwidth}\centering
                \includegraphics[width=1\columnwidth]{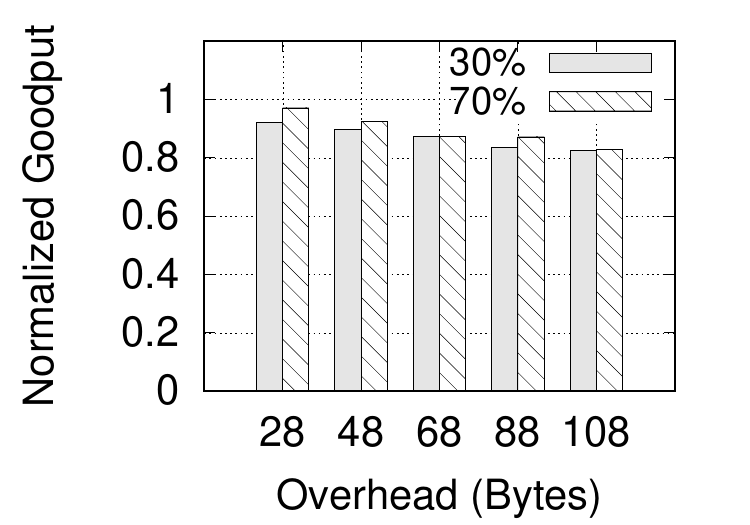}
        \end{minipage}
        \begin{minipage}{0.47\columnwidth}\centering
                \caption{Normalized average Flow Completion Time varying the network load and increasing the per-\mbox{packet overhead.}}
                \label{fig:int-fct}
        \end{minipage}\hfill\begin{minipage}{0.47\columnwidth}\centering
                \caption{Normalized average goodput of long flows (>10MB) varying the network load and increasing \mbox{per-packet overhead.}} 
                \label{fig:int-goodput}
        \end{minipage}
\end{figure}


\changed{\parab{1. High packet overheads degrade application performance.}}  The significant per-packet overheads from INT affect both flow completion time and application-level throughput, i.e., goodput. We ran an NS3~\cite{ns3} experiment to demonstrate this. We created a 5-hop fat-tree data center topology with 64 hosts connected through 10Gbps links. Each host generates traffic to randomly chosen destinations with a flow size distribution that follows a web search workload~\cite{alizadeh10}. We employed the standard ECMP routing \changed{with TCP Reno}. We ran our experiments with a range of packet overheads from 28B to 108B. The selected overheads correspond to a 5-hop topology, with one to five different INT values collected at each hop. Figure~\ref{fig:int-fct} shows the effect of increasing overheads on the average 
flow completion time (FCT) for 30\% \changed{(average)} and 70\% \changed{(high)} network utilization. 
Figure~\ref{fig:int-goodput}, 
instead, focuses on the goodput for only the long flows, i.e., with flow size >10 MBytes. Both graphs are normalized to the case where no overhead is \mbox{introduced on packets.}  
\begin{table*}[tbh]{
\begin{center}
\resizebox{1.01 \linewidth}{!}{
\begin{tabular}{lll}
\hline
\textbf{Application}                                                    & \textbf{Description}                                  & \textbf{Measurement Primitives}\\
\hline
\textbf{Per-packet aggregation}                                         &                                                       &\\
Congestion Control~\cite{katabi02,dukkipati06,han13,li19}               & Congestion Control with in-network support            & timestamp, port utilization, queue occupancy\\
Congestion Analysis~\cite{joshi18,chen19,narayana17}                    & Diagnosis of short-lived congestion events            & queue occupancy\\
Network Tomography~\cite{geng19}                                        & Determine network state, i.e., queues status          & switchID, queue occupancy\\
Power Management~\cite{heller10}                                        & Determine under-utilized network elements             & switchID, port utilization\\
Real-Time Anomaly Detection~\cite{schweller04,yu13}                     & Detect sudden changes in network status               & timestamp, port utilization, queue occupancy\\
\textbf{Static per-flow aggregation}                                    &                                                       &\\
Path Tracing~\cite{savage00,jeyakumar14,tammana16,narayana16}           & Detect the path taken by a flow or a subset           & switchID\\
Routing Misconfiguration~\cite{snoeren01,tammana16,li16}                & Identify unwanted path taken by a given flow          & switchID\\
Path Conformance~\cite{li16,tammana18,snoeren01}                        & Checks for policy violations.                         & switchID\\
\textbf{Dynamic per-flow aggregation}                                   &                                                       &\\
Utilization-aware Routing~\cite{alizadeh14,katta16,katta17}             & Load balance traffic based on network status.         & switchID, port utilization\\
Load Imbalance~\cite{li16,tammana18,savage00}                           & Determine links processing more traffic.              & switchID, port utilization\\
Network Troubleshooting~\cite{jeyakumar14,narayana17,tammana18}         & Determine flows experiencing high latency.            & switchID, timestamp\\
\hline
\end{tabular}
}
\end{center}}

\caption{Use cases enabled by PINT, organized per aggregation mode.}
\label{tab:use-cases}

\end{table*}

In the presence of 48 bytes overhead, which corresponds to 3.2\% of a 1500B packet (e.g., Ethernet has a 1500B MTU), the average FCT increases by 10\%, while the goodput for long flows degrades by 10\% if network utilization is approximately 70\%. Further increasing the overhead to 108B (7.2\% of 
a 1500B packet) leads to a 25\% increase and 20\% degradation of flow completion time and goodput, respectively. This means that even a small amount of bandwidth headroom can provide a dramatic reduction in latency~\cite{alizadeh12}.
The long flows' average goodput is approximately proportional to the residual capacity of the network. That means, at a high network utilization, the residual capacity is low, so the extra bytes in the header cause larger goodput degradation than the byte overhead itself~\cite{alizadeh12}.
As in our example, the theoretical goodput degradation should be around $1-\frac{100\%-70\%*1.072}{100\%-70\%*1.032}\approx 10.1\%$ when increasing the header overhead from 48B to 108B at around 70\% network utilization. This closely matches the experiment result, and is much larger \mbox{than the extra byte overhead (4\%).}

Although some data center networks employ jumbo frames to mitigate the problem\footnote{\scriptsize\url{ https://docs.aws.amazon.com/AWSEC2/latest/UserGuide/network_mtu.html}}, 
it is worth noting that (1) not every network can employ jumbo frames, especially the large 
number of enterprise and ISP networks; (2) some protocols might not entirely support jumbo frames; for example, RDMA over Converged Ethernet NICs provides an MTU of only 1KB~\cite{mittal18}.

\changed{\parab{2. Switch processing time.}}
In addition to consuming bandwidth, the INT overhead also affects packet processing time at switches. 
Every time a packet arrives at and departs from a switch, the bits carried over the wire need to be converted from serial to parallel and vice versa, using the 64b/66b (or 66b/64b) encoding as defined by the IEEE Standard 802.3~\cite{802.3}. 
For this reason, any additional bit added into a packet affects its processing time, delaying it at both input and output interfaces of every hop. For example, adding 48 bytes of INT data on a packet (INT header alongside two telemetry information) would cause a latency increase with respect to the original packet of almost $76ns$ and
$6ns$ for 10G and 100G interfaces, respectively\footnote{\scriptsize Consuming 48Bytes on a 10G interface requires 6 clock cycles each of them burning 6.4 ns~\cite{xilinx10g}. On a 100G interface, it needs just one clock cycle of 3ns~\cite{xilinx100g}.}. On a state-of-the-art switch with 10G interfaces, this can represent an approximately 3\% increase in processing
latency~\cite{oudin19}. On larger topologies and when more telemetry data is needed, 
the overhead on the packet can cause an increase of latency in the order of microseconds, which can \mbox{hurt the application performance~\cite{popescu17}. }

\changed{\parab{3. Collection overheads.}}
\changed{
Telemetry systems such as INT generate large amounts of traffic that may overload the network.
Additionally, INT produces reports of varying size (depending on the number of hops), while 
state-of-the-art end-host stack processing systems for telemetry data, such as Confluo~\cite{khandelwal19}, rely on \emph{fixed-byte size} headers on packets to optimize the computation overheads.
}

\section{The \sys Framework}
\label{sec:framework}

We now discuss the supported functionalities of our system, formalizing the model it works in.

\parab{Telemetry \emph{Values}.}
In our work, we refer to the telemetry information as {\em values}. Specifically, whenever a packet $p_j$ reaches a switch $s$, we assume that the switch observes a value $v(p_j,s)$. The value can be a function of the switch (e.g., port or switch ID), switch state (e.g., timestamp, latency, or queue occupancy), or any other quantity computable in the data plane. In particular, our definition supports the information types that INT~\cite{int} can collect.

\subsection{Aggregation Operations}\label{sec:agg}
We design \sys with the understanding that collecting all (per-packet per-switch) values pose an excessive and unnecessary overhead. Instead, \sys supports several aggregation operations that allow efficient encoding of the aggregated data onto packets. 
For example, congestion control algorithms that rely on the bottleneck link experienced by packets (e.g.,~\cite{li19}) can use a per-packet aggregation. Alternatively, applications that require discovering the flow's path (e.g., path conformance) can use per-flow aggregation.
\begin{itemize}[leftmargin=*]
    \item \textbf{Per-packet aggregation} summarizes the data across the different values in the packet's path, according to an \emph{aggregation function} (e.g., max/min/sum/product). For example, if the packet $p_j$ traverses the switches $s_1,s_2,\ldots,s_k$ and we perform a max-aggregation, the target quantity is $\max\set{ v(p_j,s_i)}_{i=1}^k$. 
    \item \textbf{\textit{Static} per-flow aggregation} targets summarizing values that may differ between flows or switches, but are fixed for a (flow, switch) pair. Denoting the packets of flow $x$ by $p_1,\ldots,p_z$, the static property means that for any switch $s$ on $x$'s path we have $v(p_1,s)=\ldots=v(p_z,s)$; for convenience, we denote $v(x,s)\triangleq v(p_1,s)$. If the path taken by $x$ is $s_1,\ldots, s_k$, the goal of this aggregation is then to compute all values on the path, i.e., $v(x,s_1), v(x,s_2),\ldots,v(x,s_k)$. As an example, if $v(x,s_i)$ is the ID of the switch $s_i$, then the aggregation corresponds to inferring the flow's path.
    \item \textbf{\textit{Dynamic} per-flow aggregation} summarizes, for each switch on a flow's path, the stream of values observed by its packets. Denote by $p_1,\ldots,p_z$ the packets of $x$ and by $s_1,\ldots, s_k$ its path, and let sequence of values measured by $s_{i}$ on $x$'s packets be denoted as $S_{x,i}=\langle v(p_1,s_i), v(p_2,s_i), \ldots v(p_z,s_i)\rangle$. The goal is to compute a function of $S_{i,x}$ 
    %
    according to an aggregation function (e.g., median or number of values that equal a particular value $\mathfrak v$). For example, if $v(p_j,s_i)$ is the latency of the packet $p_j$ on the switch $s_i$, using the median as an aggregation function equals computing the median latency of flow $x$ on $s_i$.
\end{itemize}

\subsection{Use Cases}\label{sec:usecases}
\sys can be used for a wide variety of use cases (see Table~\ref{tab:use-cases}). In this paper, we will mainly discuss three of them, chosen in such a way that we can demonstrate all the different \sys \mbox{aggregations in action.}

\parab{Per-packet aggregation: Congestion Control.}
State of the art congestion control solutions often use INT to collect utilization and queue occupancy statistics~\cite{li19}. 
\sys shows that we can get similar or better performance while minimizing the overheads associated with collecting the statistics.

\parab{Static per-flow aggregation: Path Tracing.}
Discovering the path taken by a flow is essential for various applications like path conformance~\cite{li16,tammana18,snoeren01}. 
In \sys, we leverage multiple packets from the same flow to infer its path.
For simplicity, we assume that each flow follows a single path.


\parab{Dynamic per-flow aggregation: Network Troubleshooting.}
For diagnosing network issues, it is useful to measure the latency quantiles from each hop~\cite{jeyakumar14,narayana17,chen19,IvkinYB019}. 
\changed{Tail quantiles are reported as the most effective way to summarize the delay in an ISP~\cite{choi2007quantile}.}
For example, we can detect network events in real-time by noticing a change in the hop latency~\cite{LatencyBarefoot}. To that end, 
we leverage \sys to collect the median and tail \mbox{latency statistics of (switch, flow) pairs.}

\subsection{Query Language}\label{sec:language}
Each query in \sys is defined as a tuple $\langle
\mbox{val\_t, agg\_t, bit-budget,}\break
\mbox{\textbf{\textit{optional}}: space-budget, flow definition, frequency}
\rangle$ that specifies\\ which values are used (e.g., switch IDs or latency), the aggregation type as in~\cref{sec:agg}, and the \emph{query bit-budget} (e.g., $8$ bits per packet). The user may also specify a space-budget that determines how much per-flow storage is allowed, the flow-definition (e.g., 5-tuple, source IP, etc.) in the case of per-flow queries,
and the query frequency (that determines which fraction of the packets \mbox{should be allocated for the query).}

\changed{\sys works with \emph{static} bit-budgets to maximize its effectiveness while remaining transparent to the sender and receiver of a packet. Intuitively, when working with INT/PINT one needs to ensure that a packet's size will not exceed the MTU even after the telemetry information is added. For example, for a 1500B network MTU, if the telemetry overhead may add to $X$ bytes, then the sender would be restricted to sending packets smaller than 1500$-X$. Thus, by fixing the budget, we allow the network flows to operate without being aware of the telemetry queries and path length.}

\subsection{Query Engine}\label{sec:engine}

\sys allows the operator to specify multiple queries that should run concurrently and a \emph{global bit-budget}. For example, if the global bit-budget is $16$ bits, we can run two $8$-bit-budget queries on the same packet.
In \sys, we add to packets a \emph{digest} -- a short bitstring whose length equals the global bit budget. This digest may compose of multiple \emph{query digests} as in the above example.

\begin{figure}
        \centering
        \includegraphics[width=\linewidth]{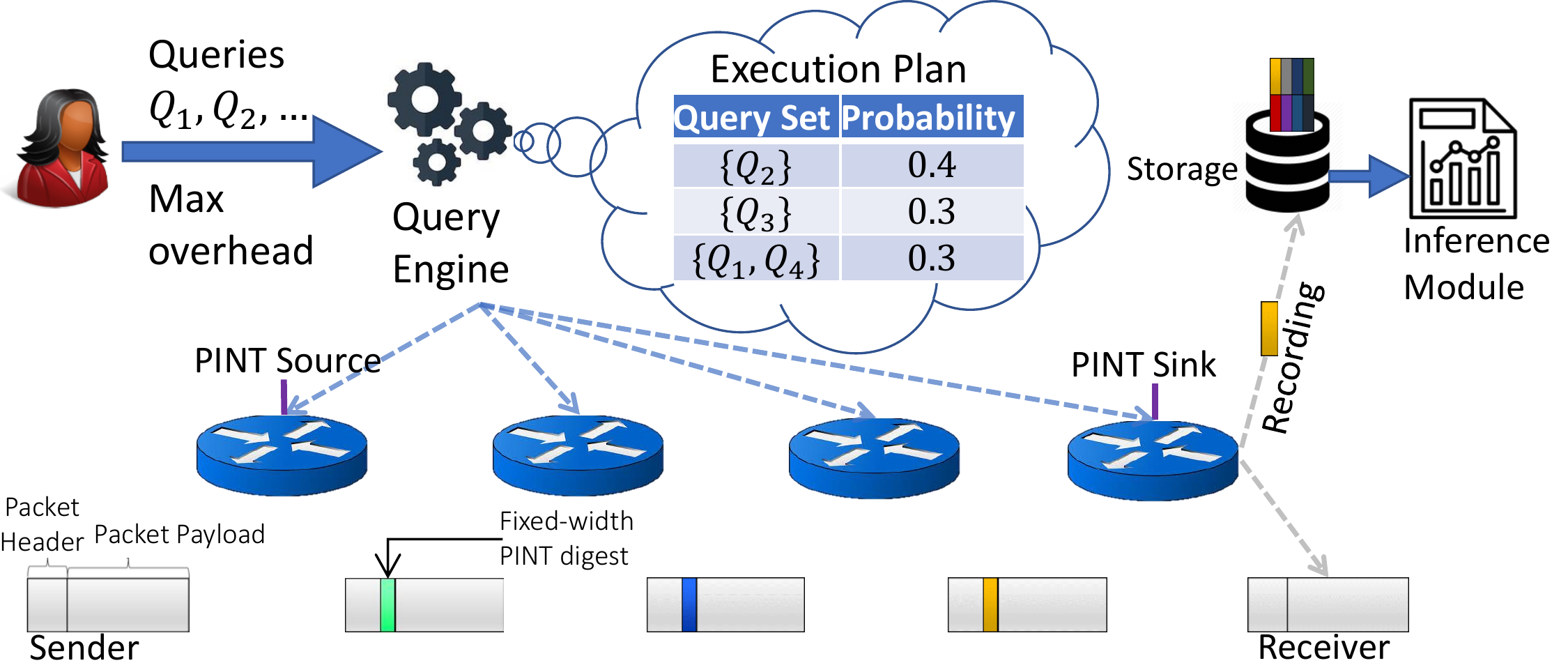}
        \caption{\small \sys's architecture: The Query Engine decides on an execution plan that determines the probability of running each query set on packets and notifies the switches. The first hop, \sys Source, adds a digest whose size is determined by the user. Every switch along the path may modify the digest but does not add bits. The last hop, \sys Sink, removes the collected telemetry information and sends it to the Recording Module. On demand, the \mbox{Inference Module is invoked to analyze the recorded data.}}
        \label{fig:arch}
        
\end{figure}

Each query instantiates an \emph{Encoding Module}, a \emph{Recording Module}, and an \emph{Inference Module}. 
The Encoding runs on the switches and modifies the packet's \emph{digest}.
When a packet reaches a \sys Sink (the last hop on its path),
the sink extracts (removes) the digest and sends the data packet to its destination. This way, \sys remains transparent to both the sender and receiver.
The extracted digest is intercepted by the Recording Module, which processes and stores the digests. 
\changed{We emphasize that the per-flow data stored by the Recording Module sits in an offline storage and no per-flow state is stored on the switches.
Another advantage of \sys is that, compared with INT, we send fewer bytes from the sink to be analyzed and thereby reduce the network overhead.} 
The Inference Module runs on a commodity server that uses the stored data to answer queries.
\mbox{\cref{fig:arch} illustrates \sys's architecture.}

Importantly, all switches must agree on which query set to run on a given packet, according to the distribution chosen by the Query Engine. We achieve coordination using a global hash function, as described in \cref{sec:globalHash}. 
Unlike INT, we do not add a telemetry header; in this way we minimize the bit overhead.\footnote{\changed{We note that removing the header is minor compared to the overhead saving \sys obtains by avoiding logging all per-hop values.}} Instead, the \sys Query Engine compiles the queries to decide on the execution plan (which is a probability distribution on a query set, see~\cref{fig:arch}) \mbox{and notifies the switches.}

\subsection{Challenges}
We now discuss several challenges we face when designing algorithms for \sys.

\parab{Bit constraints.}
In some applications, the size of values may be prohibitively large to a point where writing a single value on each packet poses an unacceptable overhead.

\parab{Switch Coordination.} 
The switches must agree on which query set to use for each packet.
While the switches can communicate by exchanging bits that are added to packets, this increases the bit-overhead of \sys and should be avoided.


\parab{Switch constraints.}
The hardware switches have constraints, including 
limited operations per packet, limited support for arithmetic operations (e.g., multiplication is not supported), inability to keep per-flow state,  etc. 
See~\cite{Precision} for a discussion of the constraints.
For \sys, these constraints mean that we must store minimal amount of state on switches and use simple encoding schemes that \mbox{adhere to the programmability restrictions.}

\section{Aggregation Techniques}
\label{sec:algorithms}
In this section, we present the techniques used by \sys to overcome the above challenges. We show how global hash functions allow efficient coordination between different switches and between switches and the Inference Module. We also show how distributed encoding schemes help reduce the number of packets needed to collect the telemetry information. Finally, we adopt compression techniques to reduce the number of bits required to represent \mbox{numeric values (e.g., latency).}

\changed{Our techniques reduce the bit-overhead on packets using probabilistic techniques. 
As a result, some of our algorithms (e.g., latency quantile estimation) are \emph{approximate}, while others (e.g., path tracing) require \emph{multiple packets} from the same flow to decode. 
Intuitively, oftentimes one mostly cares about tracing large (e.g., malicious) flows and does not require discovering the path of very short ones. Similarly, for network diagnostics it is OK to get approximated latency measurements as we usually care about large latencies or significant latency changes. We summarize which techniques apply for each of the use cases in \cref{tbl:techniques}.
}
\begin{table}[tbh]
\changed{
\begin{center}
\resizebox{1.015\linewidth}{!}{
\begin{tabular}{|c||c|c|c|}
\hline 
Use Case & Global Hashes & Distributed Coding & Value Approximation\tabularnewline
\hline 
\hline 
{Congestion Control} &\xmark&\xmark&\cmark  \tabularnewline \hline 
{Path Tracing}       &\cmark&\cmark&\xmark  \tabularnewline \hline 
{Latency Quantiles}  &\cmark&\xmark&\cmark  \tabularnewline \hline 
\end{tabular}
}
\end{center}
}
\caption{\normalfont \changed{A summary of which techniques are used for each use case.}}\label{tbl:techniques}

\end{table}

\subsection{Implicit Coordination via Global Hash Functions}\label{sec:globalHash}
In \sys, we extensively use global hash functions to determine probabilistic outcomes at the switches. As we show, this solves the switch coordination challenge, and also enables implicit coordination between switches and the Inference Module -- a feature that allows us to \mbox{develop efficient algorithms.}

\parab{Coordination among switches.} 
%
We use a global (i.e., that is known to all switches) hash function to determine which query set the current packet addresses. 
For example, suppose that we have three queries, each running with probability $1/3$, and denote the query-selection hash, mapping packet IDs to the real interval\footnote{For simplicity, we consider hashing into real numbers. In practice, we hash into $M$ bits (the range $\{0,\ldots,2^{M}-1\}$) for some integer $M$ (e.g., $M=64$). Checking if the real-valued hash is in $[a,b]$ corresponds to checking if the discrete hash is in the interval $\brackets{\floor{(2^M-1)\cdot {a}},\floor{(2^M-1)\cdot{b}}}$.} $[0,1]$, by $\mathfrak q$. 
Then if $\mathfrak q(p_j)<1/3$, all switches would run the first query, if $\mathfrak q(p_j)\in [1/3,2/3]$ the second query, and otherwise the third.
Since all switches compute the same $\mathfrak q(p_j)$, they agree on the executed query without communication.
This approach requires the ability to derive unique packet identifiers to which the hashes are applied (e.g., IPID, IP flags, IP offset, TCP sequence and ACK numbers, etc.). For a discussion on how to obtain identifiers, see~\cite{duffield01}. 

\parab{Coordination between switches and Inference Module.}
The Inference Module must know which switches modified an incoming packet's digest, but we don't want to spend bits on encoding switch IDs in the packet.  
Instead, we apply a global hash function $g$ on a (packet ID, hop number)\footnote{The hop number can be computed from the current TTL on the packet's header.} pair to choose whether to act on a packet. This enables the \sys Recording Module to compute $g$'s outcome for all hops on a packet's path and deduct where it was modified.
This coordination plays a critical role in our per-flow algorithms \mbox{as described below.}


\parab{Example \#1: Dynamic Per-flow aggregation.}
In this aggregation, we wish to collect statistics from values that vary across packets, e.g., the median latency of a (flow, switch) pair.
We formulate the general problem as follows: Fix some flow $x$.  Let $p_1,\ldots,p_z$ denote the packets of $x$ and $s_1,\ldots, s_k$ denote its path. For each switch $s_i$, we need to collect enough information about the sequence $S_{i,x}=\langle v(p_1,s_i), v(p_2,s_i), \ldots v(p_z,s_i)\rangle$ while meeting the query's bit-budget.
For simplicity of presentation, we assume that packets can store a single value.\footnote{If the global bit-budget does not allow encoding a value, we compress it at the cost of an additional error as discussed in~\cref{sec:approxVal}. If the budget allows storing multiple values, we can run the algorithm independently multiple times and thereby collect more information to improve the accuracy.}

\sys's Encoding Module runs a distributed sampling process. The goal is to have each packet carry the value of a uniformly chosen hop on the path.
That is, each packet $p_j$ should carry each value from $\set{v(p_j,s_1),\ldots,v(p_j,s_k)}$ with probability $1/k$. 
This way, with probability $1-e^{-\Omega(z/k)}$, each hop will get $z/k\cdot (1\pm o(1))$ samples, \mbox{i.e., almost an equal number.}

To get a uniform sample, we use a combination of global hashing and the Reservoir Sampling algorithm~\cite{vitter1985random}.
Specifically, when the $i$'th hop on the path (denoted $s_i$) sees a packet $p_j$, it overwrites its digest with $v(p_j,s_i)$ \emph{if $g(p_j,i) \le r_i$}. Therefore, the packet will end up carrying the value $v(p_j,s_i)$ only if (i) $g(p_j,i) \le r_i$, and (ii) $\forall \jmath\in\set{i+1,\ldots,k}:g(p_j,\jmath)> r_\jmath$. 
To get uniform sampling, we follow the Reservoir Sampling algorithm and set $r_i\triangleq1/i$. Indeed, for each hop
(i) and (ii) are simultaneously satisfied \mbox{with probability $1/k$.} 
\changed{Intuitively, while later hops have a lower chance of overriding the digest, they are also less likely to be replaced by the remaining switches along the path.}


Intuitively, we can then use existing algorithms for constructing statistics from subsampled streams. That is, for each switch $s_i$, the collected data is a uniformly subsampled stream of $S_{i,x}$.
One can then apply different aggregation functions.  For instance, we can estimate quantiles and find frequently occurring values. 
As an example, we can estimate the median and tail latency of the (flow, switch) pair by finding the relevant quantile of the subsampled stream.

On the negative side, aggregation functions like the number of distinct values or the value-frequency distribution entropy are \mbox{poorly approximable from subsampled streams~\cite{subsampled}.}

\sys aims to minimize the decoding time and amount of per-flow storage.
To that end, our Recording Module does not need to store all the incoming digests. Instead, we can use a sketching algorithm that suits the target aggregation (e.g., a quantile sketch~\cite{KLL16}). 
That is, for each switch $s_i$ through which flow $x$ is routed, we apply a sketching algorithm to the sampled substream of $S_{i,x}$.
If given a per-flow space budget (see \S\ref{sec:language}) we split it between the $k$ sketches evenly.
This allows us to record a smaller amount of per-flow information and process queries faster. 
Further, we can use a sliding-window sketch (e.g.,~\cite{arasu2004approximate,Memento,FAST}) to reflect only the most recent measurements.
Finally, the Inference Module uses the sketch to provide \mbox{estimates on the required flows. }

The accuracy of \sys for dynamic aggregation depends on the aggregation function, the number of packets $(z)$, the length of the path $(k)$, and the per-flow space stored by the Recording Module \changed{(which sits off-switch in  remote storage)}. 
We state results for two typical aggregation functions. The analysis \mbox{is deferred to~\cref{app:dynamic}.}

\begin{theorem}\label{thm:dynamicQuantiles}
Fix an error target $\epsilon\in (0,1)$ and a target quantile $\phi\in(0,1)$ (e.g., $\phi=0.5$ is the median). After seeing $O(k \epsilon^{-2})$ packets from a flow $x$, using $ O(k\epsilon^{-1})$ space, \sys produces a $(\phi\pm\epsilon)$-quantile \mbox{of $S_{x,i}$ for each hop $i$.}
\end{theorem}

\begin{theorem}\label{thm:dynamicFrequent}
Fix an error target $\epsilon\in (0,1)$  and a target threshold $\theta\in(0,1)$. After seeing $O(k \epsilon^{-2})$ packets from a flow $x$, using $O(k \epsilon^{-1})$ space, \sys produces all values that appear in at least a $\theta$-fraction of $S_{x,i}$, and no value that appears less than a $(\theta-\epsilon)$-fraction, for each hop $i$.
\end{theorem}

\subsection{Distributed Coding Schemes}
When the values are static for a given flow (i.e., do not change between packets), we can improve upon the dynamic aggregation approach using \emph{distributed encoding}. Intuitively, in such a scenario, we can spread each value $v(x,s_i)$ over multiple packets. 
The challenge is that the information collected by \sys is not known to any single entity but is rather distributed between switches. This makes it challenging to use existing encoding schemes as we wish to avoid adding extra overhead for communication between switches. Further, we need a simple encoding scheme to adhere to the switch limitations, and we desire one that allows efficient decoding.


Traditional coding schemes assume that a single encoder owns all the data that needs encoding. However, in \sys, the data we wish to collect can be distributed among the network switches. That is, the message we need to transfer is partitioned between the \mbox{different switches along the flow's path.} 


\begin{figure}
        \includegraphics[width=\linewidth]{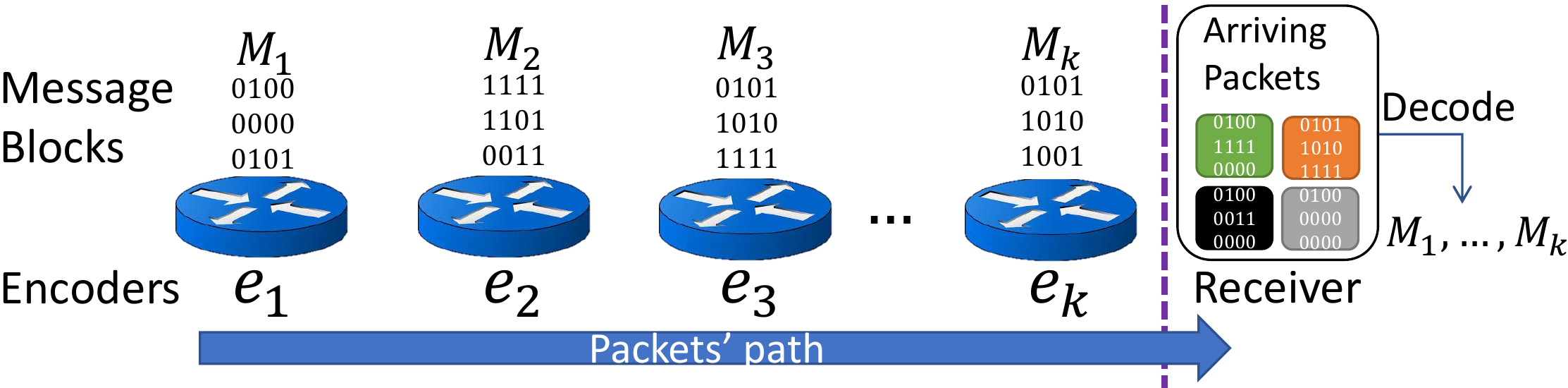}
        
        \caption{\small \mbox{Multiple encoders send a distributed message.}
        }
        
        \label{fig:encoding}
        
\end{figure}

We present an encoding scheme that is fully distributed without any communication between encoders.
Specifically, we define our scheme as follows:
a \emph{sequence} of $k$ encoders hold a $k$-block message $M_1,\ldots,M_k$ such that encoder $e_i$ has $M_i$ for all $i\in\set{1,\ldots,k}$.
The setting is illustrated in \cref{fig:encoding}.
Each packet carries a digest which has a number of bits that equals the block size and has a unique identifier which distinguishes it from other packets. Additionally, each encoder is aware of its hop number (e.g., by computing it from the TTL field in the packet header).
The packet starts with a digest of $\overline 0$ (a zero bitstring) and passes through $e_1,\ldots,e_k$. Each encoder can modify the packet's digest before passing it to the next encoder. After the packet visits $e_k$, it is passed to the Receiver, which tries to decode the message. We assume that the encoders are \emph{stateless} to model the switches' inability to keep a per-flow state in networks.

Our main result is a distributed encoding scheme that needs $k\cdot\log\log^* k\cdot(1+o(1))$ packets for decoding the message with near-linear decoding time. We note that Network Coding~\cite{ho2003benefits} can also be adapted to this setting. However,  we have found it rather inefficient, as we explain later on.

\begin{figure}

\centering
\subfigure[Algorithm Progress]{\label{fig:b}\includegraphics[width=.49\columnwidth]{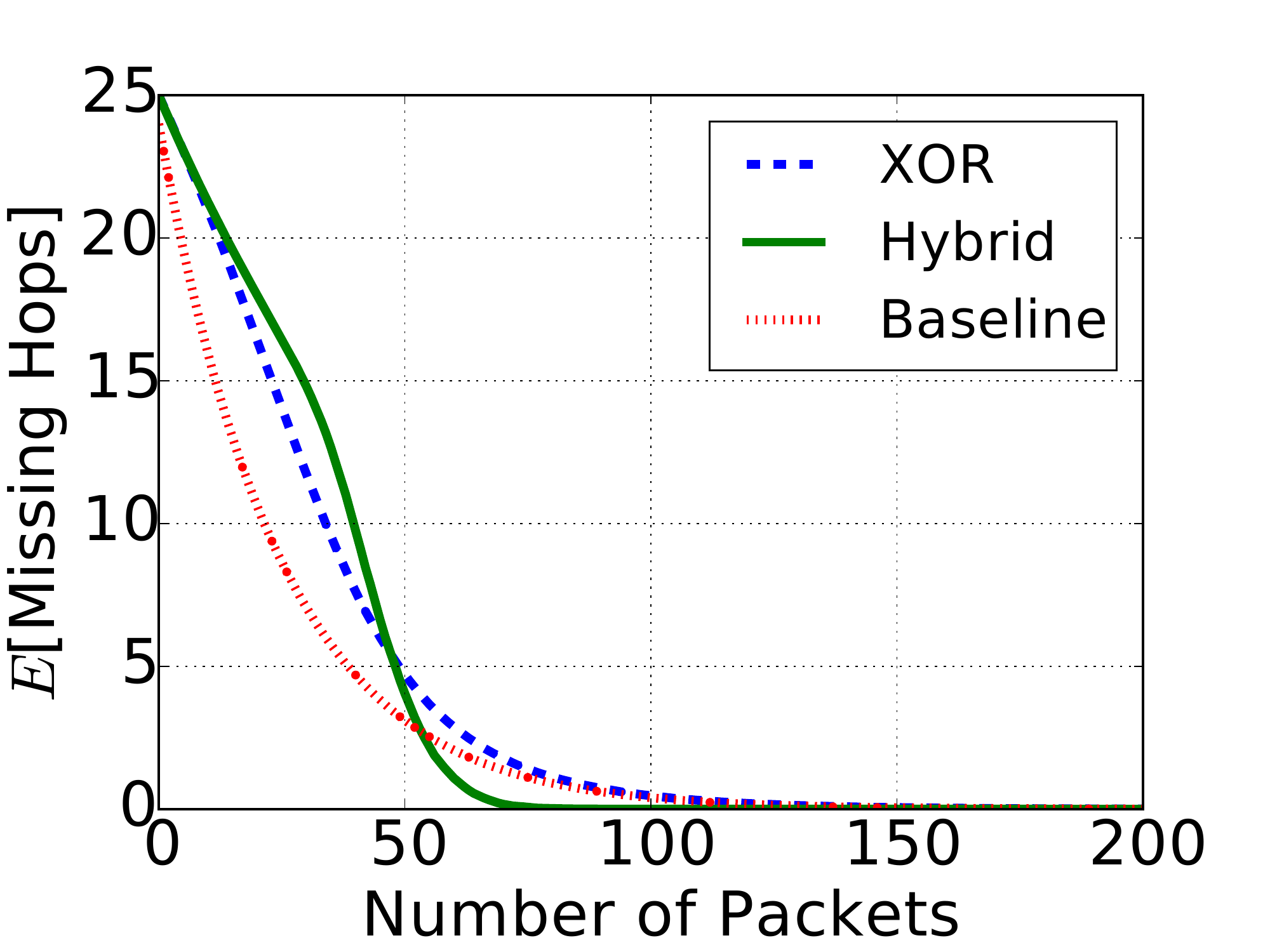}}
\subfigure[Probability of Decoding]{\label{fig:a}\includegraphics[width=.49\columnwidth]{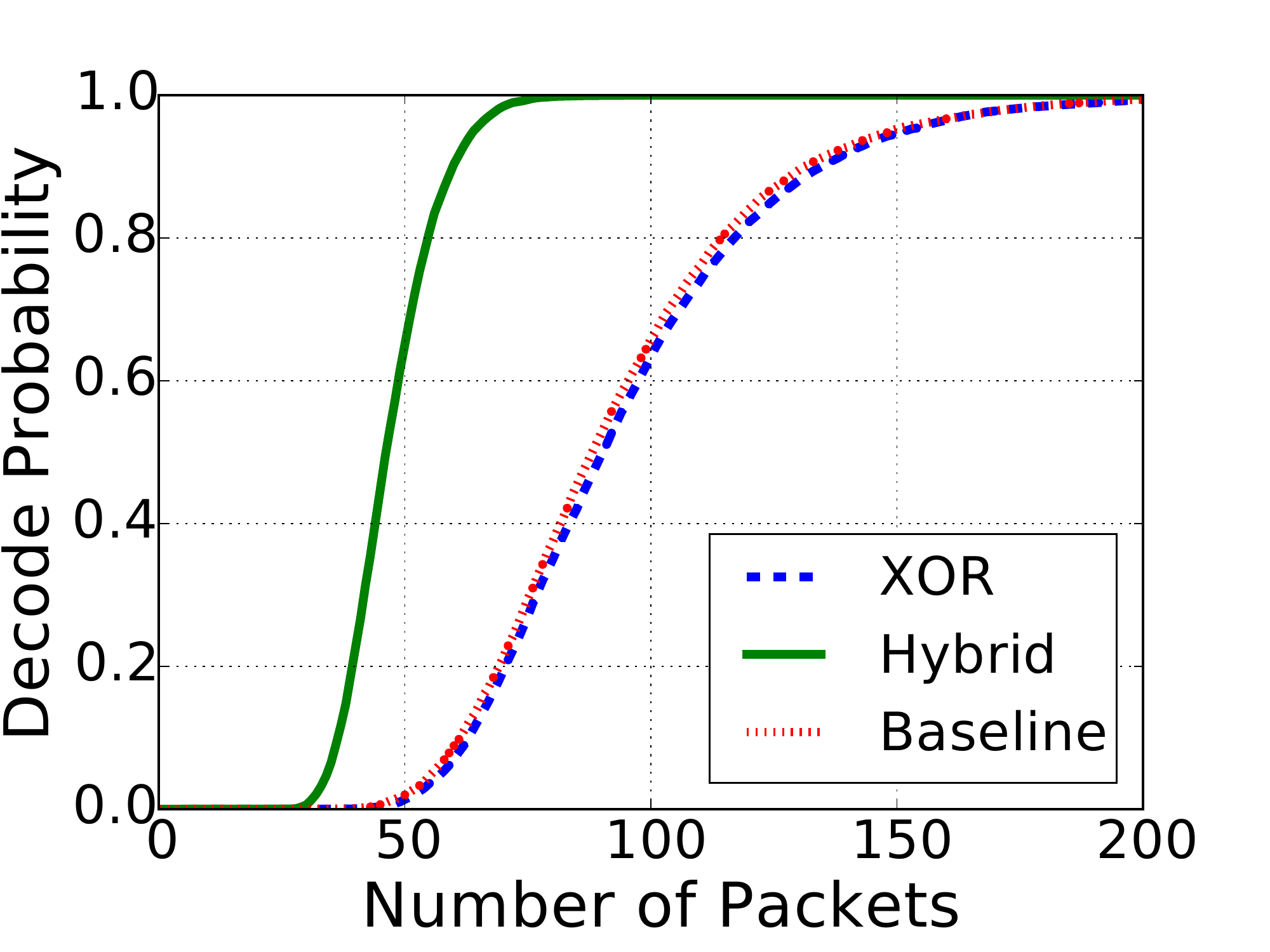}}

\caption{The XOR scheme (with prob. $1/d$) decodes fewer hops at first but is able to infer the entire path 
using a similar number of packets to Baseline.
By interleaving both schemes (Hybrid), we get a better result as the first hops are mainly decoded by Baseline packets and the last hops by XOR packets that have XOR probability $\log\log d/\log d$ and are more likely to hit the missing hops. Plotted for $d=k=25$ hops.}
\label{fig:hybrid}
\end{figure}

\parab{Baseline Encoding Scheme.}
A simple and intuitive idea for a distributed encoding scheme is to carry a uniformly sampled block on each packet. That is, the encoders can run the Reservoir Sampling algorithm using a global hash function to determine whether to write their block onto the packet. Similarly to our Dynamic Aggregation algorithm, the Receiver can determine the hop number of the sampling switch, by evaluating the hash function, and report the message.

The number of packets needed for decoding the message using this scheme follows the Coupon Collector Process (e.g., see~\cite{flajolet1992birthday}), where each block is a coupon and each packet carries a random sample. It is well-known that for $k$ coupons, we would need $k\ln k(1+o(1))$ samples on average to collect them all. 
For example, for $k=25$, Coupon Collector has a median (i.e., probability of 50\% to decode) of $89$ packets and a 99'th percentile of $189$ packets, \mbox{as shown in \cref{fig:hybrid}.}

The problem with the Baseline scheme is that while the first blocks are encoded swiftly, later ones require a higher number of packets. The reason is that after seeing most blocks, every consecutive packet is unlikely to carry a new block. 
\changed{This is because the encoders are unaware of which blocks were collected and the probability of carrying a new block is proportional to number of missing blocks.}
As a result, the Baseline scheme has a long ``tail'', meaning that completing the decoding requires many packets.


\parab{Distributed XOR Encoding.}
An alternative to the Baseline scheme is to use bitwise-xor while encoding. 
We avoid assuming that the encoders know $k$, but assume that they know a \emph{typical length} $d$, such that $d=\Theta(k)$. Such an assumption is justified in most cases; for example, in data center topologies we often know a tight bound on the number of hops~\cite{tammana16}. Alternatively, the median hop count in the Internet is estimated to be $12$~\cite{HopCount2}, while only a few paths have more than $30$ hops~\cite{theilmann2000dynamic,carter1997server}.
The XOR encoding scheme has a parameter $p$, and each encoder on the path bitwise-xors its message onto the packet's digest with probability $p=1/d$, according to the global hash function. That is, the $i$'th encoder changes the digest if $g(p_j,i)<p$. We note that this probability is uniform and that the decision of whether to xor is independent for each encoder, allowing a distributed implementation without communication \mbox{between the encoders.}

When a packet reaches the Receiver, the digest is a bitwise-xor of multiple blocks $M_{i_1}\oplus\ldots\oplus M_{i_K}$, where $K$ is a binomial random variable $K\sim \mathit{Bin}(k,p)$. 
The Receiver computes $g(p_j,1),\ldots,g(p_j,k)$ to determine the values $i_1,\ldots,i_K$. If this set contains exactly one unknown message block, we can discover it by bitwise-xoring the other blocks. For example, if we have learned the values of $M_1,M_3,M_4,M_6$ and the current digest is $p_j.\mbox{dig}=M_1\oplus M_5\oplus M_6$, we can derive $M_5$ since $M_5=p_j.\mbox{dig}\oplus M_1\oplus M_6$.

On its own, the XOR encoding does not asymptotically improve over the Baseline. Its performance is optimized when $p=1/d=\Theta(1/k)$, where it requires $O(k\log k)$ packets to decode, i.e., within a constant factor from the Baseline's performance. 
Interestingly, we show that the combination of the two approaches gives better results.

\parab{Interleaving the Encoding Schemes.}
Intuitively, the XOR and Baseline schemes behave differently. In the Baseline, the chance of learning the value of a message block with each additional packet \emph{decreases} as we receive more blocks. In contrast, to recover data from an XOR packet, we need to know all xor-ed blocks but one. When $p$ is much larger than $1/k$, many packet digests are modified by multiple encoders, which means that the probability to learn a message block value \emph{increases} \mbox{as we decode more blocks.} 

\changed{
As an example for how the interleaved scheme helps, consider the case of $k=2$ encoders.
The Baseline scheme requires three packets to decode the message in expectation; the first packet always carries an unknown block, but each additional packet carries the missing block with probability only $1/2$.
In contrast, suppose each packet chooses the Baseline scheme and the XOR scheme each with probability $1/2$, using $p=1$. For the interleaved scheme to complete, we need either two Baseline packets that carry different blocks or one XOR packet and one Baseline packet. A simple calculation shows that this requires just $8/3$ packets in expectation.
}

For combining the schemes, we first choose whether to run the Baseline with probability $\tau$, or the XOR otherwise. Once again, switches make the decision based on a global hash function applied to the packet identifier to achieve implicit agreement on the packet type. Intuitively, the Baseline scheme should reduce the number of undecoded blocks from $k$ to $k'$, and the XOR will decode the rest. To minimize the number of packets, we can set $\tau=3/4$ and the XOR probability\footnote{\mbox{If $d\le 15$ then $\log\log d<1$; in this case we set the probability to $1/\log d$.}} to $\log\log d/\log d$ to reduce the required number of packets to $O(k\log\log k / \log\log\log k$). In such setting, the Baseline decodes most hops, leaving $k'\approx k/\log k$ for the XOR layer.
For example, when $k=25$, we get a median of $41$ packets and a 99'th percentile of $68$ packets to decode the message. That is, not only does it improve the average case, the interleaving has sharper tail bounds. This improvement is illustrated in~\cref{fig:hybrid}.

\parab{Multi-layer Encoding.} So far, we used a single probability for xor-ing each packet, which was chosen inversely proportional to $k'$ (the number of hops that were not decoded by the Baseline scheme). This way, we maximized the probability that a packet is xor-ed by exactly one of these $k'$ blocks, and we xor any block from the $k-k'$ that are known already to remove them from the decoding.
However, when most of the $k'$ blocks left for XOR are decoded, it also ``slows down'' and requires more packets for decoding each additional block. Therefore, we propose to use multiple XOR \emph{layers} that vary in their sampling probabilities. We call the Baseline scheme layer $0$, and the XOR layers $1,\ldots,\mathcal L$. Each XOR layer $\ell\in\set{1,\ldots,\mathcal L}$ starts with $k_\ell$ undecoded blocks, xors with probability $p_\ell$, and \mbox{ends when $k_{\ell+1}$ blocks are undecoded.}


Our analysis, given in~\cref{app:static}, shows that by optimizing the algorithm parameters $\tau, \mathcal L, \set{k_\ell}_{\ell=1}^\mathcal L$ and $\set{p_\ell}_{\ell=1}^\mathcal L$, we obtain the following result. 
The value of $\mathcal L$ is a function of $d$, and we have that $\mathcal L=1$ if $d\le \floor{e^e}=15$ and $\mathcal L=2$ if $16\le d\le e^{e^e}$; i.e., \mbox{in practice we need only one or two XOR layers.}

\begin{theorem}
After seeing $k\log\log^* k (1+o(1))$ packets, the Multi-layer scheme
can decode the message. 
\end{theorem}

We note that the $o(1)$ term hides an $O(k)$ packets additive term, where the constant depends on how well $d$ approximates $k$. Namely, when $d=k$, our analysis indicates that $k(\log\log^* k + 2 + o(1))$ packets are enough.
Finally, we note that if $d$ is not representative of $k$ at all, we still get that $k\ln k(1+o(1))$ packets are enough, the same as in the Baseline scheme (up to lower order terms). The reason is that our choice of $\tau$ is close to $1$, i.e., only a small fraction of the packets are used in the XOR layers.

\parab{Comparison with Linear Network Coding.}
Several algorithms can be adapted to work in the distributed encoding setting. For example, Linear Network Coding (LNC)~\cite{ho2003benefits} allows one to decode a message in a near-optimal number of packets by taking random linear combinations over the message blocks. That is, on every packet, each block is xor-ed into its digest with probability $1/2$. Using global hash functions to select which blocks to xor, one can determine the blocks that were xor-ed onto each digest. LNC requires just $\approx k + \log_2 k$ packets to decode the message. However, in some cases, LNC may be suboptimal and \sys can use alternative solutions. First, the LNC decoding algorithm requires matrix inversion which generally takes $O(k^3)$ time in practice (although theoretically faster algorithms are possible). 
If the number of blocks is large, we may opt for approaches with faster decoding.
Second, LNC does not seem to work when using hashing to reduce the overhead.
As a result, in such a setting, LNC could use fragmentation, but may require a larger number of packets than the \mbox{XOR-based scheme using hashing.}


\parab{Example \#2: Static Per-flow Aggregation.}
We now discuss how to adapt our distributed encoding scheme for \sys's static aggregation. 
Specifically, we present solutions that allow us to reduce the overhead on packets to meet the bit-budget in case a single value cannot be written on a packet. For example, for determining a flow's path, the values may be $32$-bit switch IDs, while the bit-budget can be smaller (even a single bit per packet). 
We also present an implementation variant that allows to decode the collection of packets in near-linear time. This improves the quadratic time required for computing $\set{g(p_j,i)}$ for all packets $p_j$ and hops $i$.

\parab{Reducing the Bit-overhead using Fragmentation.}
Consider a scenario where each value has $\mathfrak q$ bits while we are allowed to have smaller $\mathfrak b$-bits digests on packets.
In such a case, we can break each value into $F\triangleq\ceil{\mathfrak q/\mathfrak b}$ \emph{fragments} where each has $\le \mathfrak b$ bits. Using an additional global hash function, each packet $p_j$ is associated with a \emph{fragment number} in $\set{1,\ldots,F}$. 
We can then apply our distributed encoding scheme separately on each fragment number. While fragmentation reduces the bit overhead, it also increases the number of packets required for the aggregation, and the decode complexity, as if there were $k\cdot F$ hops.

\parab{Reducing the Bit-overhead using Hashing.}
The increase in the required number of packets and decoding time when using fragmentation may be prohibitive in some applications. 
We now propose an alternative that allows decoding with fewer packets, if the value-set is restricted. Suppose that we know in advance a small set of possible block values $\mathcal V$, such that any $M_i$ is in $\mathcal V$. For example, when determining a flow's path, $\mathcal V$ can be the set of switch IDs in the network. Intuitively, the gain comes from the fact that the keys may be longer than $\log_2 |\mathcal V |$ bits (e.g., switch IDs are often 32-bit long, while networks have much fewer than $2^{32}$ switches).
Instead of fragmenting the values to meet the $\mathfrak q$-bits query bit budget, we leverage hashing.
Specifically, we use another global hash function $h$ that maps (value, packet ID) pairs into $\mathfrak q$-bit bitstrings. When encoder $e_i$ sees a packet $p_j$, if it needs to act it uses $h(M_i,p_j)$ to modify the digest. In the Baseline scheme $e_i$ will write $h(M_i,p_j)$ on $p_j$, and in the XOR scheme it will xor $h(M_i,p_j)$ onto its current digest. 
As before, the Recording Module checks the hop numbers that modified the packet. 
The difference is in how the Inference Module works -- for each hop number $i$, we wish to find a single value $v\in\mathcal V$ that agrees with all the Baseline packets from hop $i$. For example, if $p_1$ and $p_2$ were Baseline packets from hop $i$, $M_i$ must be a value such that $h(M_i,p_1)=p_1.\mbox{dig}$ and $h(M_i,p_2)=p_2.\mbox{dig}$.
If there is more than one such value, the inference for the hop is not complete and we require additional packets to determine it. Once a value of a block $M_i$ is determined, from any digest $p_j$ that was xor-ed by the $i$'th encoder, we xor $h(M_i,p_j)$ from $p_j.\mbox{dig}$. This way, the number of unknown blocks whose hashes xor-ed $p_j$ decreases by one. If only one block remains, we can treat it similarly to a Baseline packet and use it to reduce the number {of potential values for that block.}
\changed{Another advantage of the hashing technique is that it does not assume anything about the width of the values (e.g., switch IDs), as long as each is distinct.}

\parab{Reducing the Decoding Complexity.}
Our description of the encoding and decoding process thus far requires processing is super-quadratic ($\omega(k^2)$) in $k$. That is because we need $\approx k\log\log^* k$ packets to decode the message, and we spend $O(k)$ time per packet in computing the $g$ function to determine which encoders modified its digest.
We now present a variant that reduces the processing time to nearly linear in $k$. 
Intuitively, since the probability of changing a packet is $\Omega(1/k)$, the number of random bits needed to determine which encoders modify it is $O(k\log k)$. Previously, each encoder used the global function $g$ to get $O(\log k)$ pseudo-random bits and decide whether to change the packet. Instead, we can use $g$ to create $O(\log 1/p)=O(\log k)$ pseudo-random $k$-bit vectors. Intuitively, each bit in the bitwise-and of these vectors will be set with probability $p$ (as defined by the relevant XOR layer).
The $i$'th encoder will modify the packet if the $i$'th bit is set in the bitwise-and of the vectors\footnote{This assumes that the probability is a power of two, or provides a $\sqrt 2$ approximation of it. By repeating the process we can get a better approximation.}.
At the Recording Module, we can compute the set of encoders that modify a packet in time $O(\log k)$ by drawing the random bits and using their bitwise-and. Once we obtain the bitwise-and vector we can extract a list of set bits in time $O(\#\mbox{set bits})$ using bitwise operations. Since the average number of set bits is $O(1)$, the overall per-packet complexity remains $O(\log k)$ and the total decoding time becomes $O(k\log k\log\log^* k)$. We note that this improvement assumes that $k$ fits in $O(1)$ machine words (e.g., $k\le 256$) and that encoders can do $O(\log k)$ \mbox{operations per packet.}

\parab{Improving Performance via Multiple Instantiations.}
The number of packets \sys needs to decode the message depends on the query's bit-budget. 
However, increasing the number of bits in the hash may not be the best way to reduce the required number of packets. Instead, we can use multiple independent repetitions of the algorithm. For example, given an $8$-bit query budget, we can use two independent $4$-bit hashes.



\subsection{\changed{Approximating Numeric Values}}\label{sec:approxVal}
Encoding an exact numeric value on packet may require too many bits, imposing an undesirable overhead. For example, the 32-bit latency measurements that INT collects may exceed the bit-budget. We now discuss to compress the value, at the cost of \mbox{introducing an error.}

\parab{Multiplicative approximation.} One approach to reducing the number of bits 
required to encode 
a value is to write on the packet's digest $\mathfrak a(p_j,s)\triangleq \brackets{\log_{(1+\epsilon)^2}{v(p_j,s)}}$ instead of $v(p_j,s)$. Here, the $\brackets{\cdot}$ operator rounds the quantity to the closest integer. At the Inference Module, we can derive a $(1+\epsilon)$-approximation of the original value by computing $(1+\epsilon)^{2\cdot\mathfrak a(p_j,s)}$.
For example, if we want to compress a $32$-bit value into $16$ bits, \mbox{we can set $\epsilon=0.0025$.}

\parab{Additive approximation.} If distinguishing small values is not as crucial as bounding the maximal error, we obtain better results by encoding the value with additive error instead of multiplicative error. For a given error target $\Delta$ (thereby reducing the overhead by $\floor{\log_2\Delta}$ bits), the Encoding Module writes $\mathfrak a(p_j,s)\triangleq \brackets{\frac{v(p_j,s)}{2\Delta}}$, and the Inference Module computes $(2\Delta)\cdot{\mathfrak a(p_j,s)}$. 

\parab{Randomized counting.} For some aggregation functions, the aggregation result may require more bits than encoding a single value. For example, in a per-packet aggregation over a $k$-hop path with $q$-bit values, the sum may require $q+\log k$ bits to write explicitly while the product may take $q\cdot k$ bits. This problem is especially evident if $q$ is small (e.g., a single bit specifying whether the latency is high).
Instead, we can take a randomized approach to increase the value written on a packet probabilistically. For example, we can estimate the number of high-latency hops or the end-to-end latency to within a $(1+\epsilon)$-multiplicative factor using \mbox{$O(\log\epsilon^{-1}+\log\log (2^q \cdot k\cdot \epsilon^2)))$ bits~\cite{ApproximateCounting}.}

\stepcounter{figure}
\begin{figure*}
\centering
\subfigure[\changed{Web search workload (large flows)}]{\label{fig:goodput_vs_load}\includegraphics[width=0.33\linewidth]{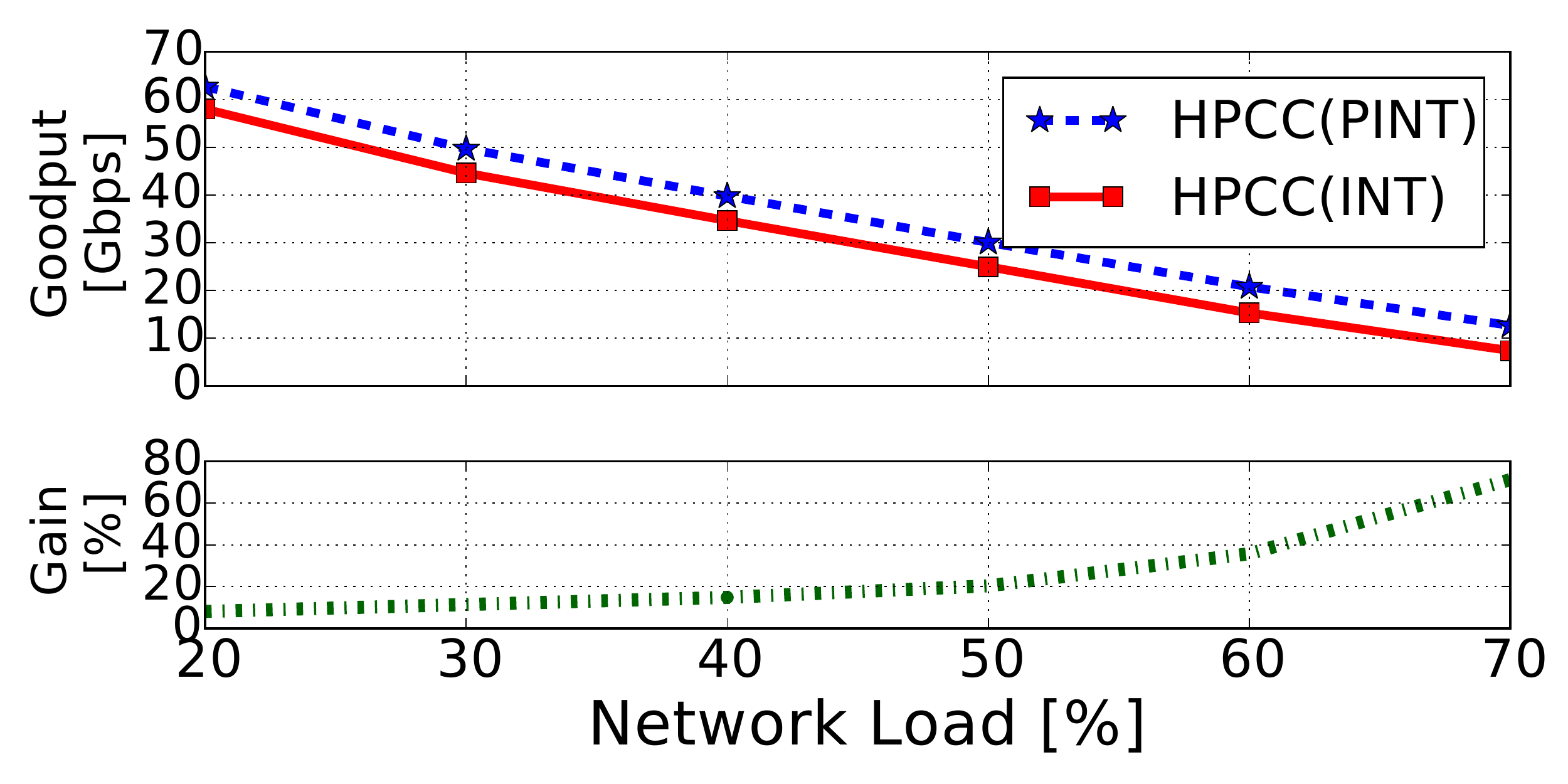}}
\hspace*{-3mm}
\subfigure[\changed{Web search workload}]{\label{fig:HPCCvsPINT-wb}
\includegraphics[width=0.33\linewidth]{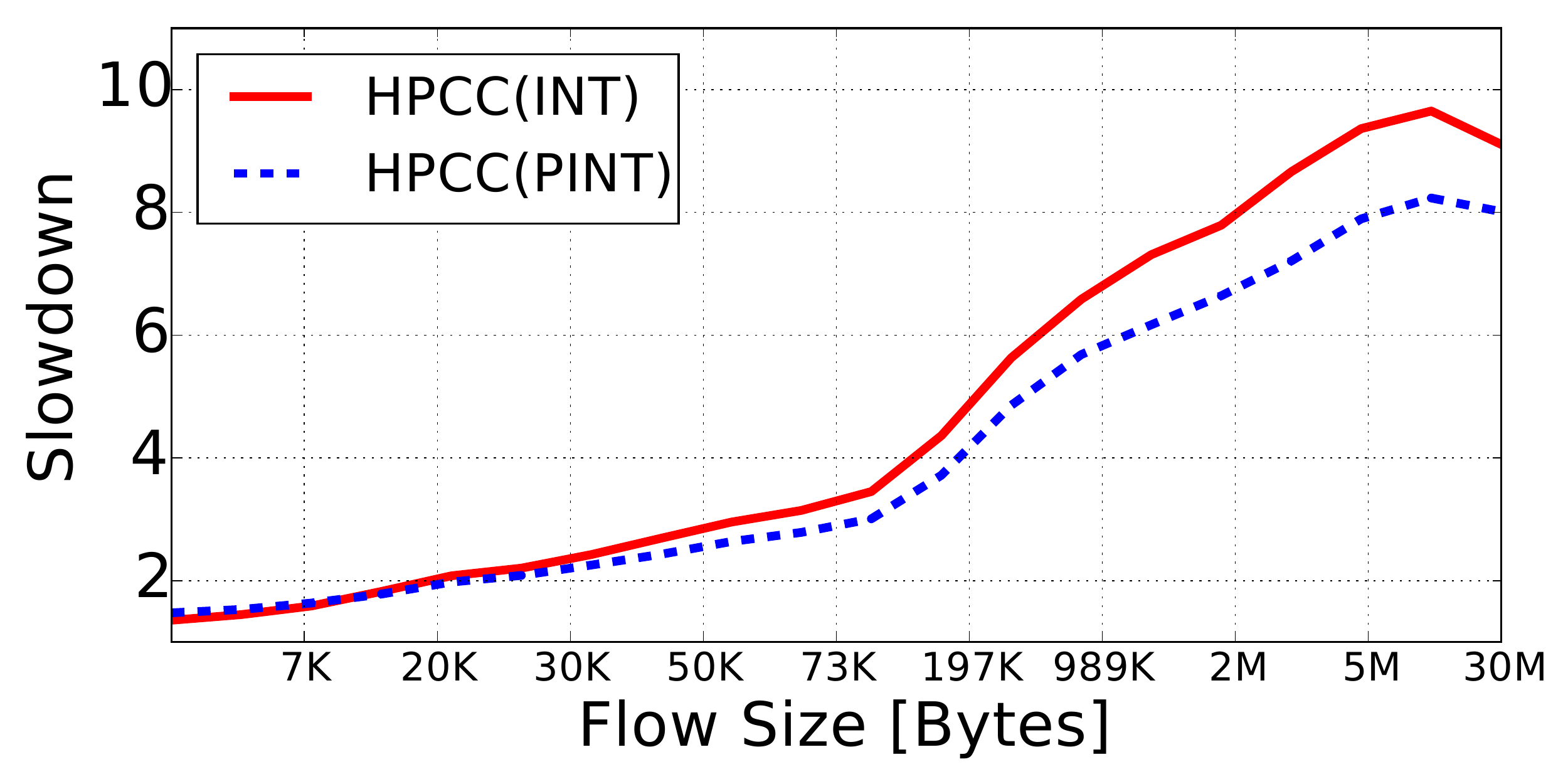}
}
\hspace*{-3mm}
\subfigure[\changed{Hadoop workload}]{\label{fig:HPCCvsPINT-fb}\includegraphics[width=0.33\linewidth]{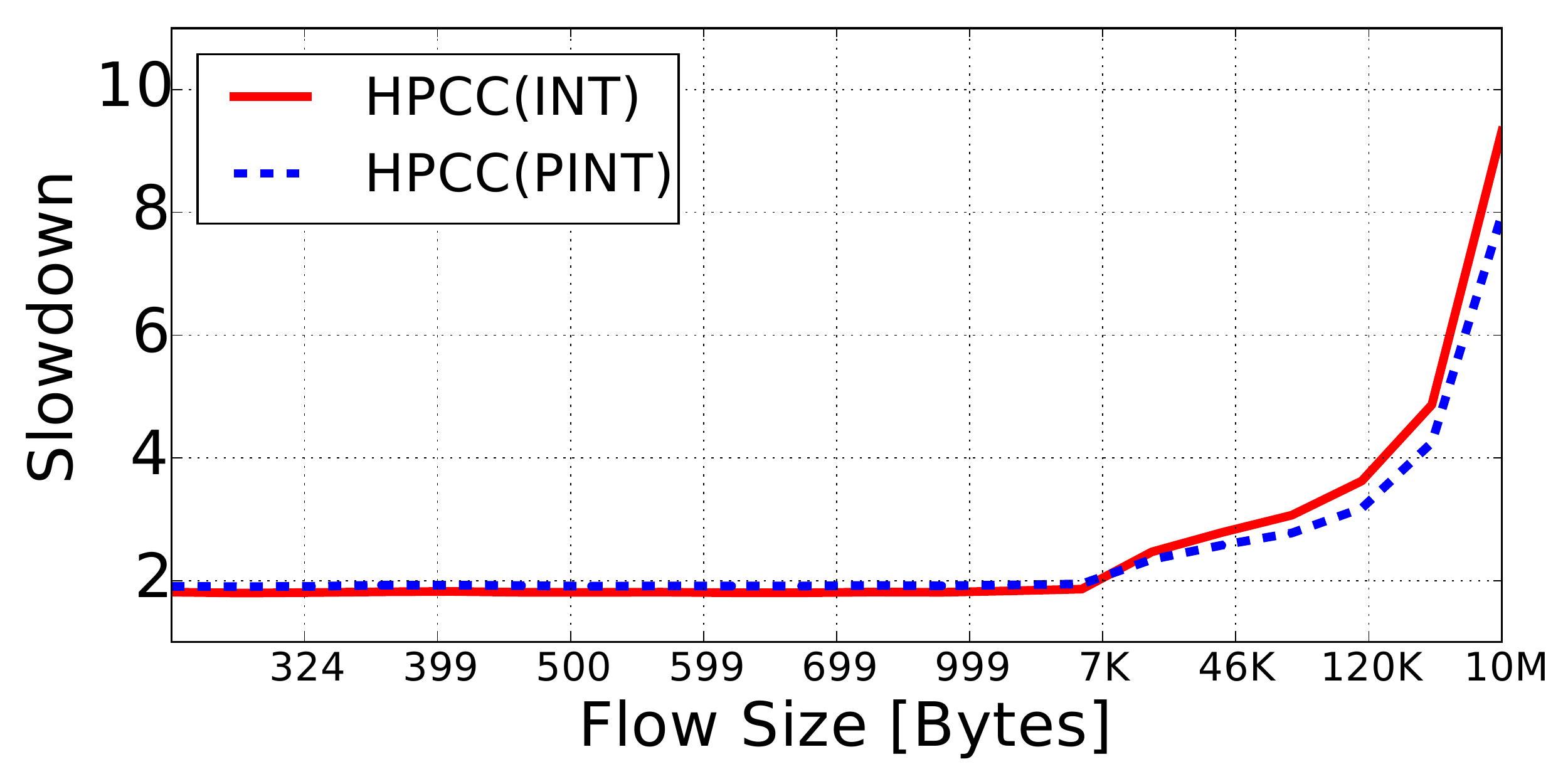}}

\caption{\small \changed{Comparison of the 95th-percentile slowdown of the standard INT-based HPCC and the \sys-based HPCC. \sys improves the performance for the long flows due to its reduced overheads. In (b) and (c), the network load is 50\% and the x-axis scale is chosen such that there are 10\% of the flows between consecutive tick marks.}}

\end{figure*} 
\addtocounter{figure}{-2}

\begin{figure}
        \includegraphics[width=.95\linewidth]{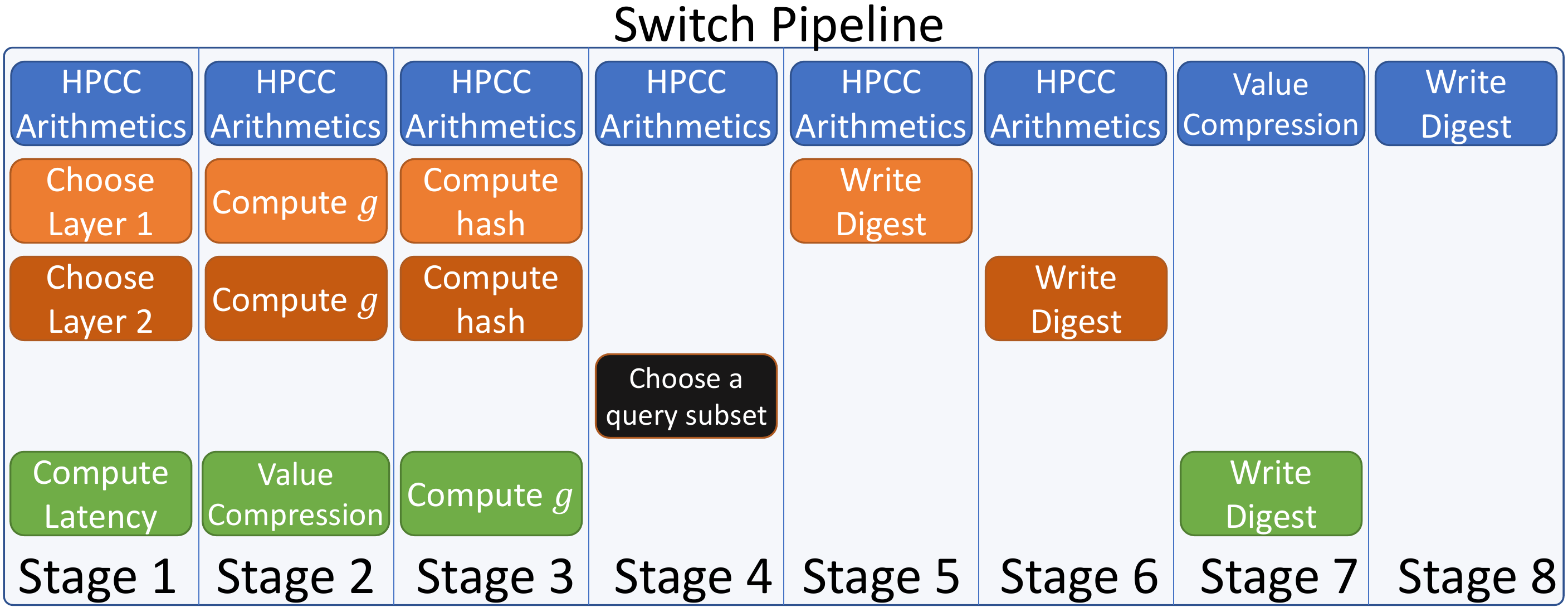}
        
        \caption{\small {Layout illustration for two path tracing hashes, alongside a latency query, \mbox{and a congestion control query.}}
        }
        \label{fig:layout}
        
\end{figure}
\stepcounter{figure}
\parab{Example \#3: Per-packet aggregation.}
Here, we wish to summarize the data across the different values in the packet's path. For example, HPCC~\cite{li19} collects per-switch information carried by INT data, and adjusts the rate at the end host according to the highest link utilization along the path. 
To support HPCC with \sys, we have two key insights: (1) we just need to keep the highest utilization (i.e., the bottleneck) in the packet header, instead of every hop; (2) we can use the multiplicative approximation to further reduce the number of bits for storing the utilization.
\changed{Intuitively, \sys improves HPCC as it reduces the overheads added to packets, as explained in \cref{sec:motiv}.}

In each switch, we calculate the utilization as in HPCC, with slight tuning to be supported by switches (discussed later). The multiplication is calculated using log and exp based on lookup tables \cite{sharma17}. The result is encoded using multiplicative approximation. To further eliminate systematic error, we write $\mathfrak a(p_j,s)\triangleq \brackets{\log_{(1+\epsilon)^2}{v(p_j,s)}}_R$, the $\brackets{\cdot}_R$ \emph{randomly} performs floor or ceiling, with a probability distribution that gives an expected value equals to $\log_{(1+\epsilon)^2}{v(p_j,s)}$. This way, some packets will overestimate the utilization while others underestimate it, thus resulting in the correct value on average.
In practice, we just need 8 bits to support $\epsilon=0.025$.

\parab{Tuning HPCC calculation for switch computation.}
We maintain the exponential weighted moving average (EWMA) of link utilization $U$ of each link in the switch. $U$ is updated on every packet with:
$U=\frac{T-\tau}{T}\cdot U+\frac{\tau}{T}\cdot u$, where
$u=\frac{\mathit{qlen}}{B\cdot T}+\frac{\mathit{byte}}{B\cdot\tau}$ is the new sample for updating $U$.
Here, $T$ is the base RTT and $B$ is the link bandwidth (both are constants).
Intuitively, the weight of the EWMA, $\frac{\tau}{T}$,  corresponds to each new packet's time occupation $\tau$. The calculation of $u$ also corresponds to each new packet: $\mathit{byte}$ is the packet's size, and $\mathit{qlen}$ is the queue length when the packet is dequeued\footnote{This is slightly different from HPCC, where the calculation is done in the host, which can only see packets of its own flow. Therefore, the update is scaled for packets of the same flow ($\tau$ is time gap between packets of the same flow, and $\mathit{byte}$ includes the bytes from other flows in between). Here, the update is performed on all packets on the same link. Since different flows may interleave on the link, \mbox{our calculation is more fine-grained.}}.

To calculate the multiplications, we first do the following transformation: $U=\frac{T-\tau}{T}\cdot U+\frac{\mathit{qlen}\cdot\tau}{B\cdot T^2}+\frac{\mathit{byte}}{B\cdot T}$.
Then we calculate the multiplications using logarithm and exponentiation as detailed in~\cref{app:hpcc}.

\section{Implementation}
\label{sec:impl}
\sys is implemented using the P4 language and can be deployed on commodity programmable switches.
We explain how each of our use cases is executed.

For running the path tracing application (static per-flow aggregation), we require four pipeline stages.
The first chooses a layer, another computes $g$, the third hashes the switch ID to meet the query's bit budget, and the last writes the digest.
If we use more than one hash for the query, both can be executed in parallel as \mbox{they are independent.}

Computing the median/tail latency (dynamic per-flow aggregation) also requires four pipeline stages: one for computing the latency, one for compressing it to meet the bit budget; one to compute $g$; and one \mbox{to overwrite the value if needed.  }


Our adaptation of the HPCC congestion control algorithm requires six pipeline stages to compute the link utilization, followed by a stage for approximating the value and another to write the digest.
\changed{For completeness, we elaborate on how to implement in the data plane the different arithmetic operations needed by HPCC in \cref{app:arithmeticImplementation}. We further note that running it may require that the switch would need to perform the update of $U$ in a single stage. In other cases, we propose to store the last $n$ values of $U$ on separate stages and update them in a round-robin manner, for some integer $n$. This would mean that our algorithm would need to recirculate every $n$'th packet as the switch's pipeline is one-directional.}

Since the switches have a limited number of pipeline stages, we parallelize the processing of queries as they are independent of each other.
We illustrate this parallelism for a combination of the three use cases of \sys.
We start by executing all queries simultaneously, writing their results on the packet vector. Since HPCC requires more stages than the other use cases, we concurrently compute which query subset to run according to the distribution selected by the Query Engine (see \S\ref{sec:engine}). We can then write the digests of all the selected queries without increasing the number of stages compared with running HPCC alone. The switch layout for such a combination is \mbox{illustrated in \cref{fig:layout}.}

%
%
%
%
%
%
%
%

\section{Evaluation}\label{sec:eval}
We evaluate on the three use cases discussed on \S\ref{sec:usecases}.

\subsection{Congestion Control}
We evaluate \changed{how \sys affects the performance of HPCC~\cite{li19}} using the same simulation setting as in~\cite{li19}. 
\changed{Our goal is not to propose a new congestion control scheme, but rather to present a low-overhead approach for collecting the information that HPCC utilizes.}
We use NS3~\cite{ns3} and a FatTree topology with 16 Core switches, 20 Agg switches, 20 ToRs, and 320 servers (16 in
each rack).  Each server has a single 100Gbps NIC connected to a single ToR. The capacity of each link between Core and Agg switches, as well as Agg switches and ToRs, are all 400Gbps. All links have a 1$\mu$s propagation delay, which gives a 12$\mu$s maximum base RTT. The switch buffer size is 32MB. 
The traffic is generated following the flow size distribution in \changed{web search from Microsoft \cite{alizadeh10} and} Hadoop from Facebook \cite{roy15}. Each server generates new flows according to a Poisson process, destined to random servers. The average flow arrival time is set so that the total network load is 50\% \changed{(not including the header bytes)}. 
We use the recommended setting for HPCC: $W_{AI}=80$ bytes, \mbox{$maxStage=0$, $\eta=95\%$, and $T=13\mu$s.}

\newcommand{\smallFig}{}
\ifdefined\smallFig
\begin{figure}
\centering
\hspace*{-2mm}
\subfigure[\changed{Web search workload }]{\label{fig:probHPCC}\includegraphics[width=0.512\linewidth]{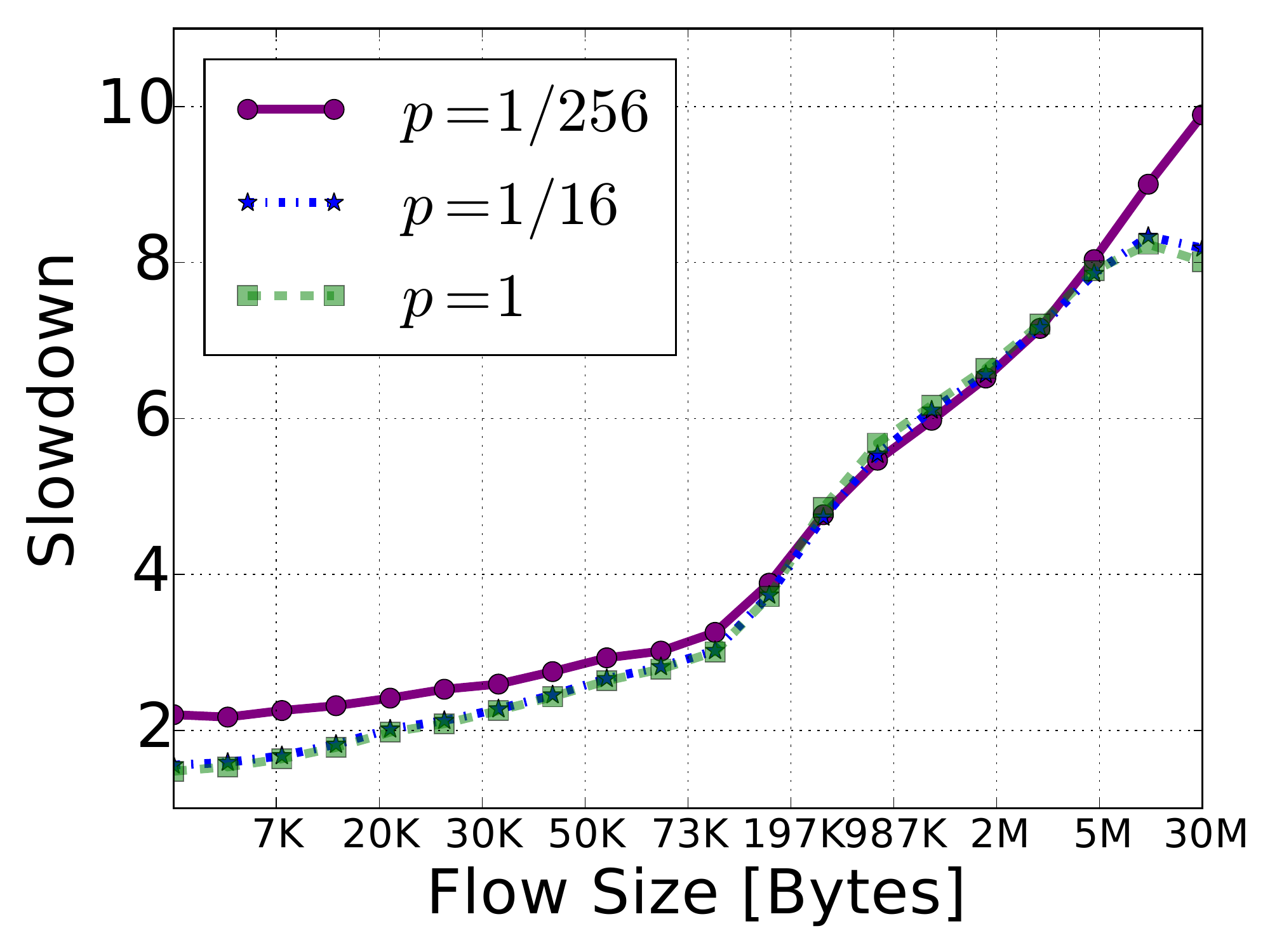}}
\hspace*{-2mm}
\subfigure[\changed{Hadoop workload }]{\label{fig:probHPCCHadoop}\includegraphics[width=0.512\linewidth]{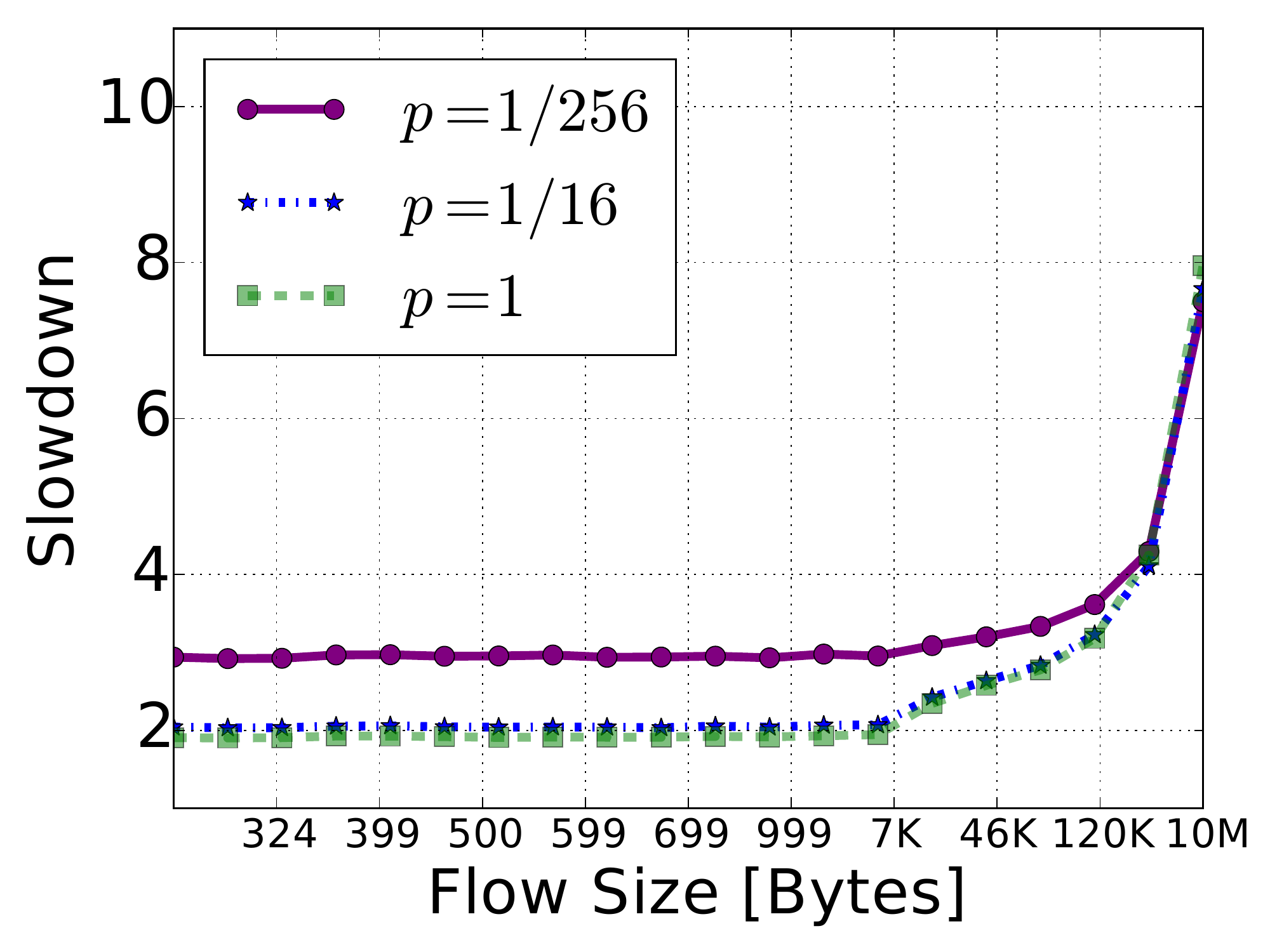}}
\hspace*{-4mm}

\caption{\small \changed{The 95th-percentile slowdown of running \sys-based HPCC (at 50\% network load) on $p$-fraction of the packets. On both workloads, the performance of running it on $1/16$ of the packets produces similar results to running it on all.}}
\label{fig:probPint}

\end{figure}
\else
\begin{figure*}
\centering
\subfigure[\changed{Web search workload}]{\label{fig:probHPCC}
\includegraphics[width=0.49\linewidth]{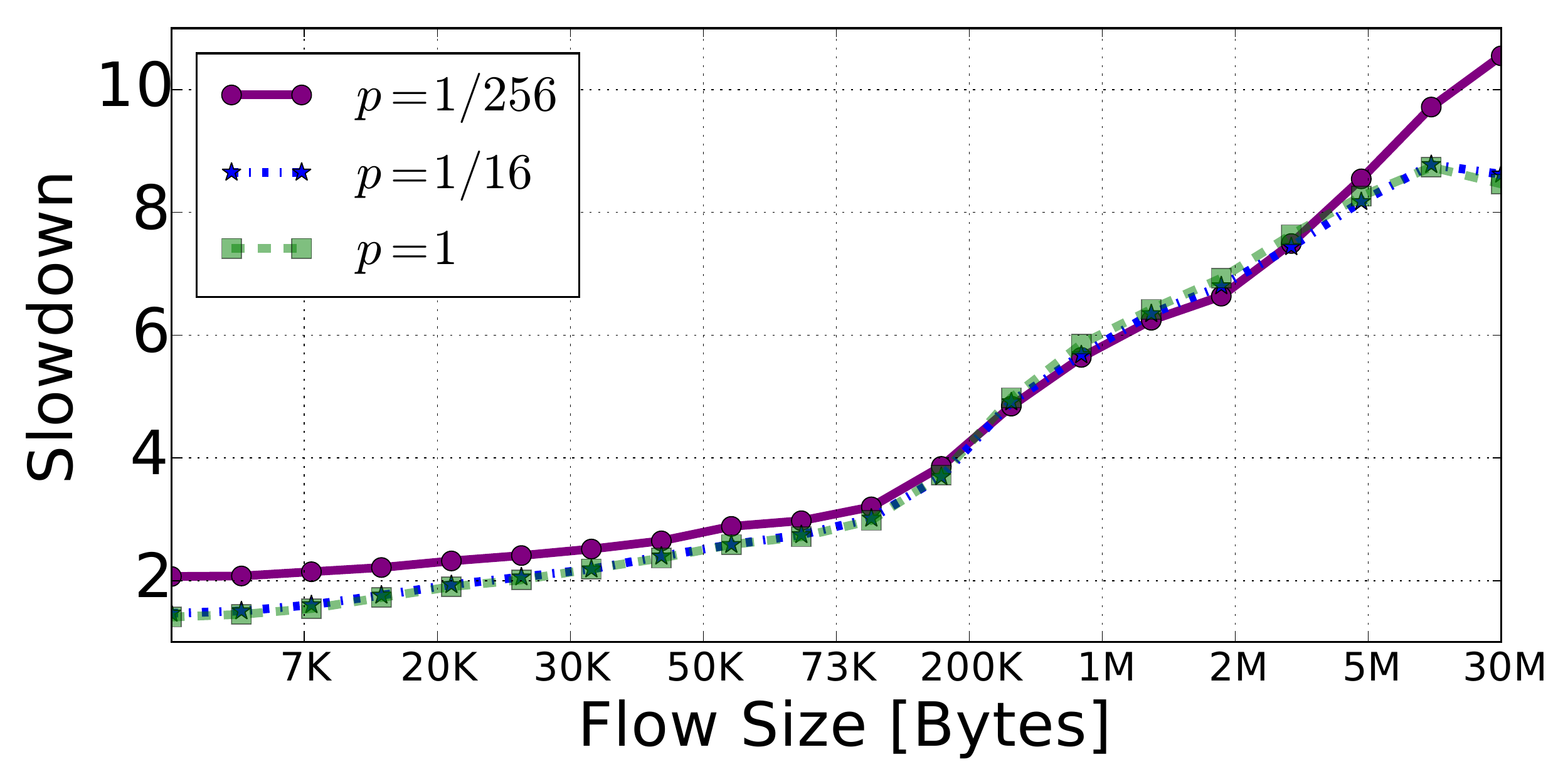}
}
\subfigure[\changed{Hadoop workload}]{\label{fig:probHPCCHadoop}\includegraphics[width=0.49\linewidth]{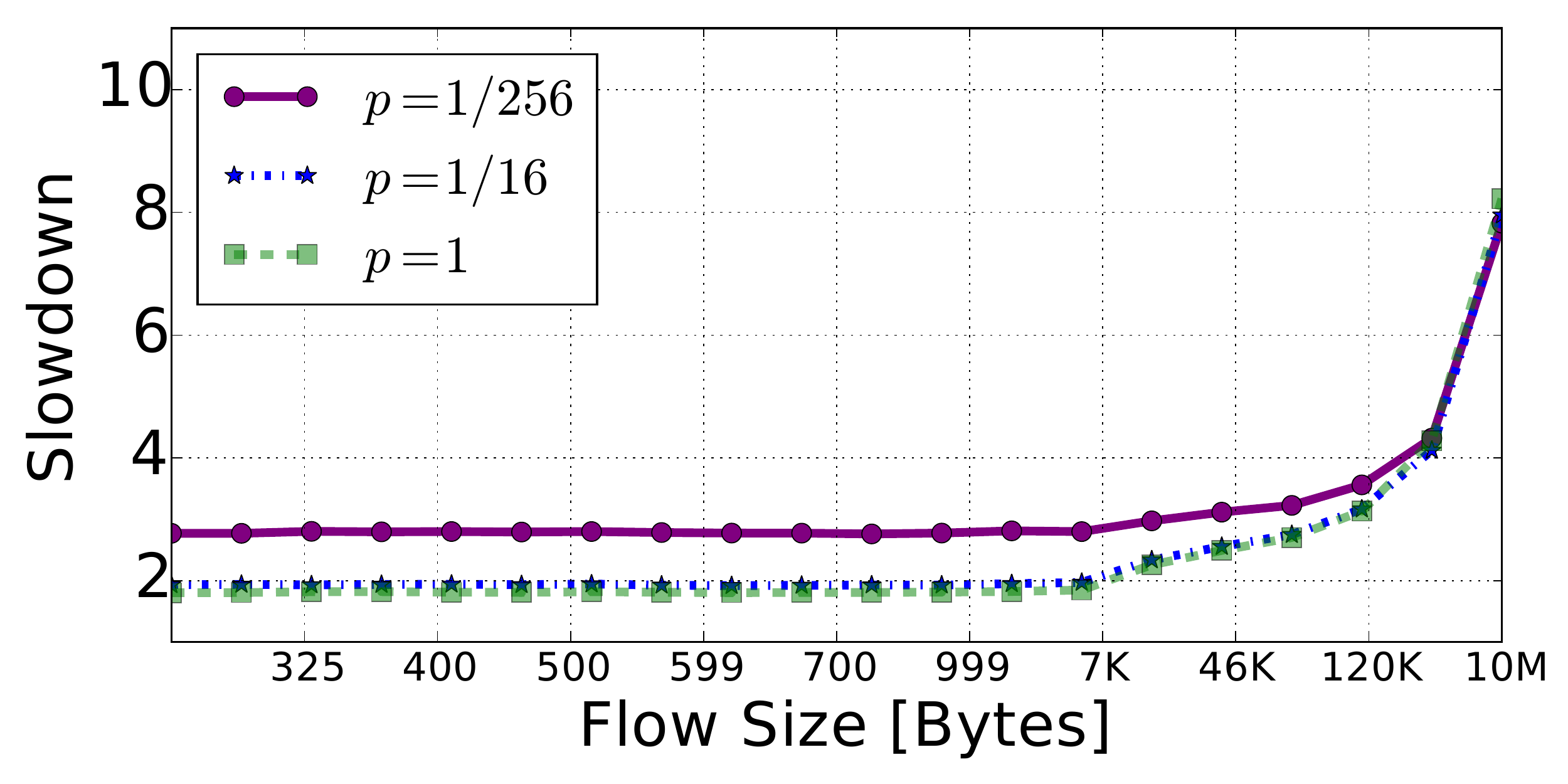}}
\caption{\small \changed{The 95th-percentile slowdown of running congestion control on $p$-fraction of the packets. On both workloads, the performance of running it on $1/16$ of the packets produces similar results to running PINT on all.}}
\label{fig:probPint}
\end{figure*}
\fi

The results, depicted in~\cref{fig:HPCCvsPINT-wb} and \cref{fig:HPCCvsPINT-fb}, show that \sys has similar performance \changed{(in terms of slowdown)} to HPCC, despite using just $8$ bits per packet. 
\changed{Here, \emph{slowdown} refers to the ratio between the completion time of the flow in the presence of other flows and alone.} 
Specifically, \sys has better performance on long flows while slightly worse performance on short ones. The better performance on long flows is due to \sys's bandwidth saving. 
\changed{\cref{fig:goodput_vs_load} shows the relative goodput improvement,  averaged over all flows over 10MB, of using \sys at different network load. At higher load, the byte saving of \sys brings more significant improvement. For example, at 70\% load, using \sys improves the goodput by 71\%. This trend aligns with our observation in \S \ref{sec:motiv}.}

To evaluate how the congestion control algorithm would perform alongside other queries, we experiment in a setting where only a 
$p=1,1/16,1/256$ 
fraction of the packets carry the query's digest. \changed{As shown in~\cref{fig:probHPCC} and~\cref{fig:probHPCCHadoop}}, the performance 
only slightly degrades for $p=1/16$. \changed{This is expected, because the bandwidth-delay product (BDP) is 150 packets, so there are still 9.4 ($\approx$150/16) packets per RTT carrying feedback. Thus the rate is adjusted on average once per 1/9.4 RTT (as compared to 1/150 RTT with per-packet feedback), which is still very frequent.}
\changed{\sout{This suggests that congestion control algorithms may not need per-packet feedback and that multiple feedback results per sending window are enough. }}
With $p=1/256$, the performance of short flows degrades significantly, \changed{because it takes longer than an RTT to get feedback.} \changed{\sout{ This is because there are at most 150 packets in an RTT, so $p=1/256$ may give no feedback for an RTT.}} The implication is that congestion caused by long flows is resolved slowly, so the queue lasts longer, resulting in higher latency for short flows. 
\changed{The very long flows (e.g., > 5MB) also have worse performance. The reason is that they are long enough to collide with many shorter flows, so when the competing shorter flows finish, the long flows have to converge back to the full line rate. With $p=1/256$, it takes much longer time to converge than with smaller $p$.}

\changed{In principle, the lower feedback frequency $p$ only affects the convergence speed as discussed above, but not the stability and fairness. Stability is guaranteed by no overreaction, and HPCC's design of reference window (constant over an RTT) provides this regardless of $p$.
Fairness is guaranteed by additive-increase-multiplicative-decrease (AIMD), which is preserved regardless of $p$.
}

\begin{figure}
\centering

\includegraphics[width=1.01\columnwidth]{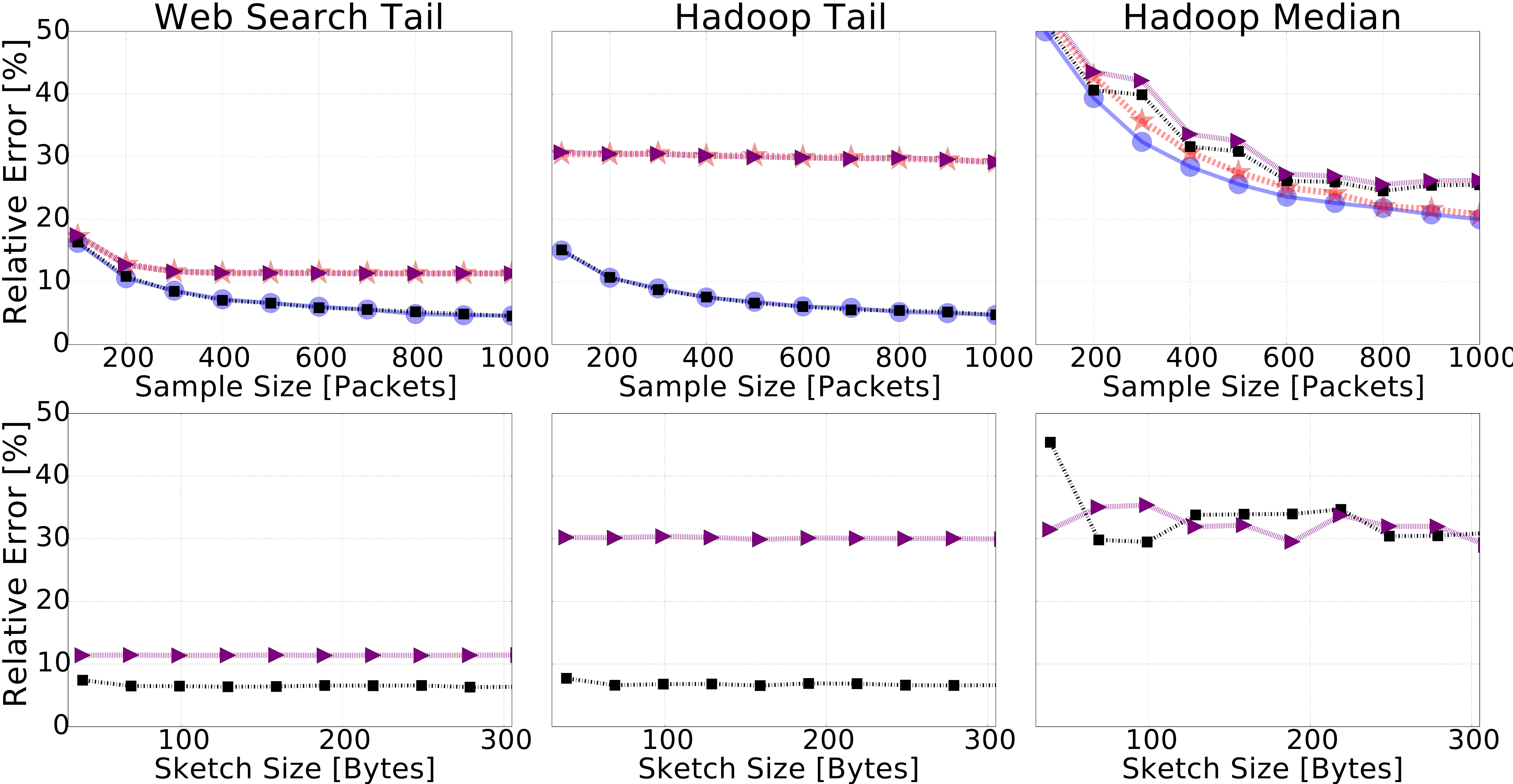}
\includegraphics[width=1.01\columnwidth]{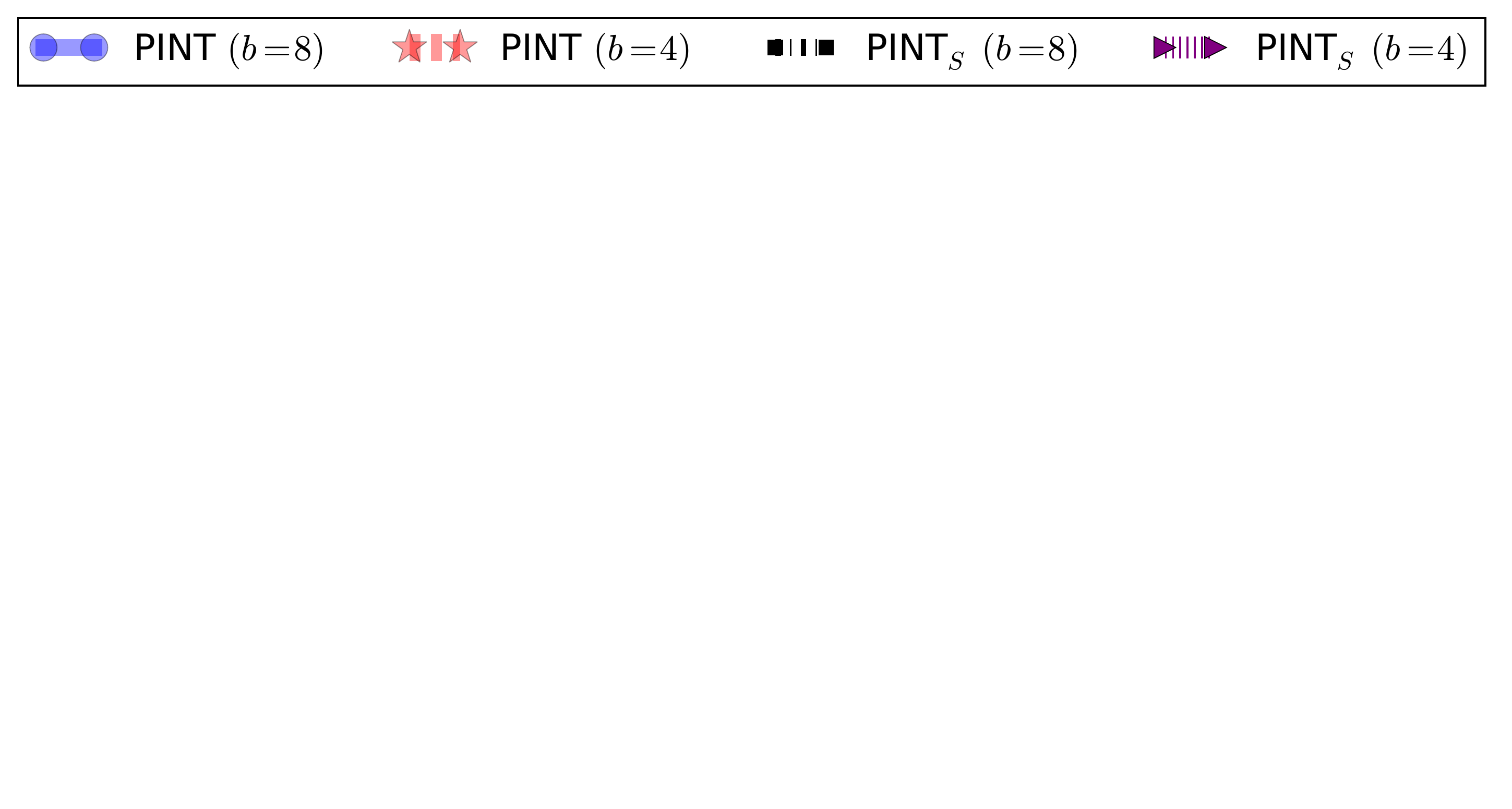}
\vspace*{-48mm}
\caption{\small \changed{\sys error on estimating latency quantiles with a sketch ({\normalfont PINT$_S$}) and without. In the first row, the sketch has $100$ digests; in the second, the sample has $500$ packets.}}
\label{fig:latency}
\end{figure}

\begin{figure*}
\centering
\subfigure[Kentucky Datalink ($D=59$)]{\label{fig:d59_avg}\includegraphics[width=.33\linewidth]{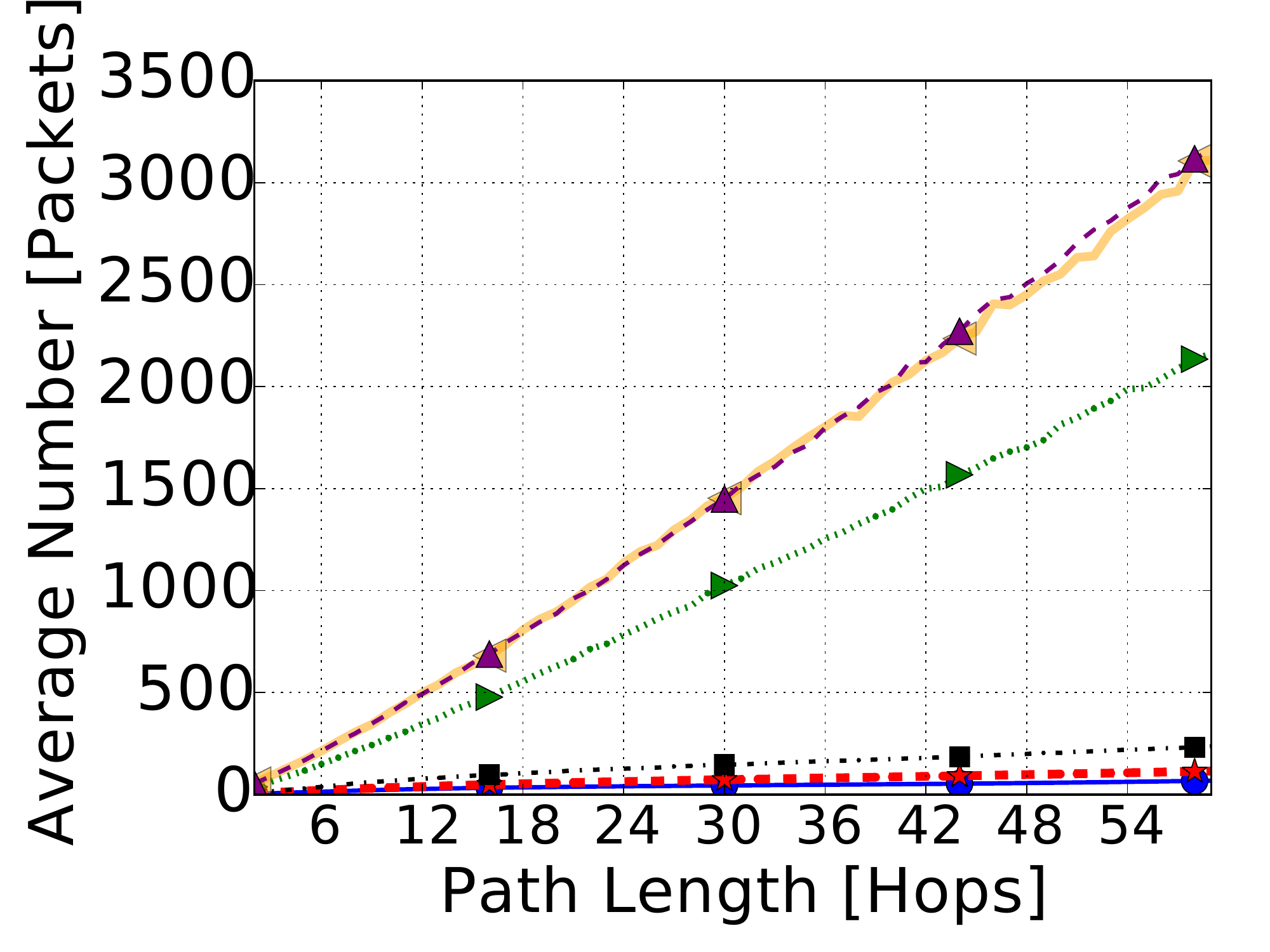}}
\subfigure[US Carrier ($D=36$)]{\label{fig:d36_avg}\includegraphics[width=.33\linewidth]{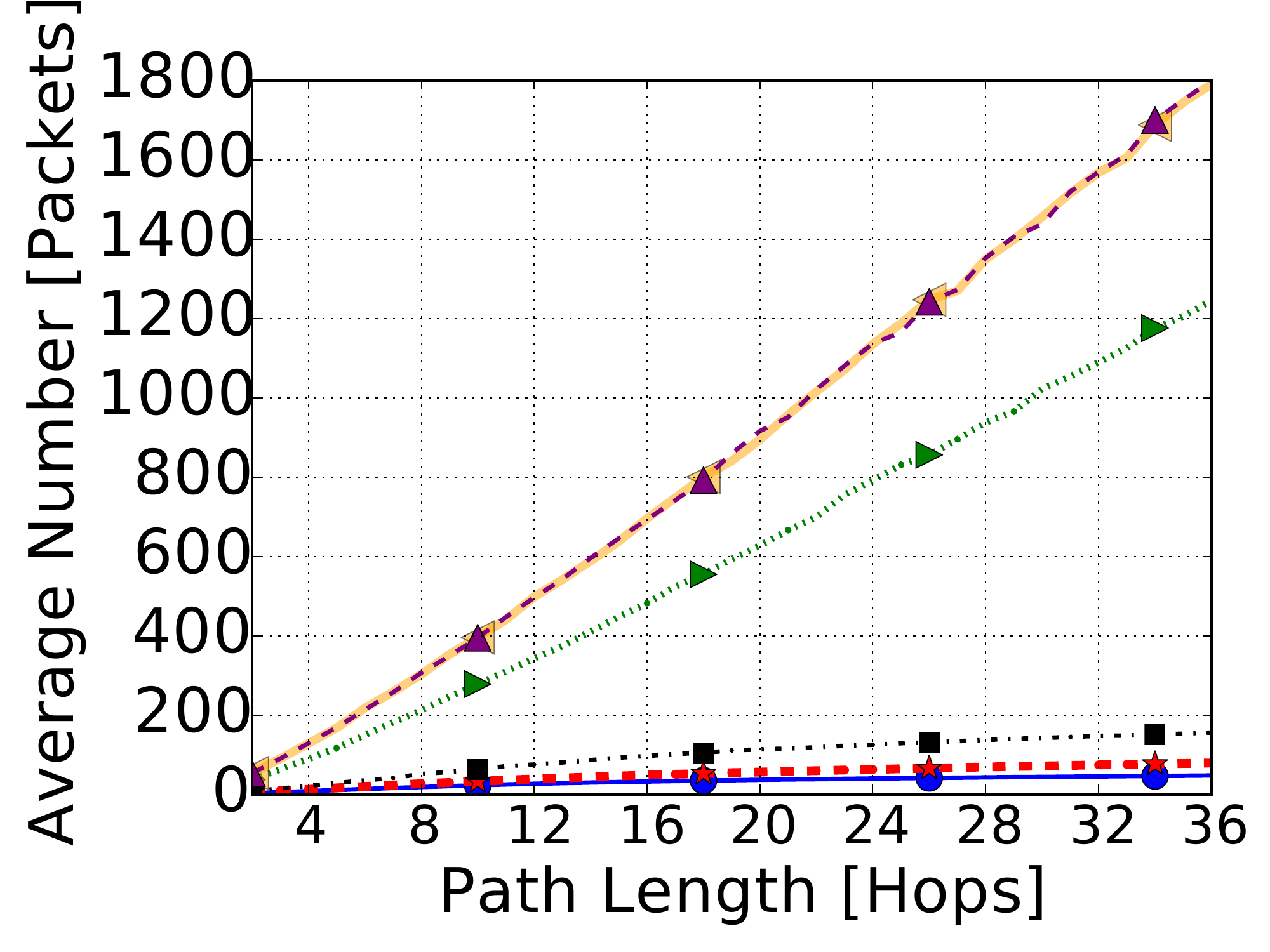}}
\subfigure[Fat Tree ($D=5$)]{\label{fig:ft_avg}\includegraphics[width=.33\linewidth]{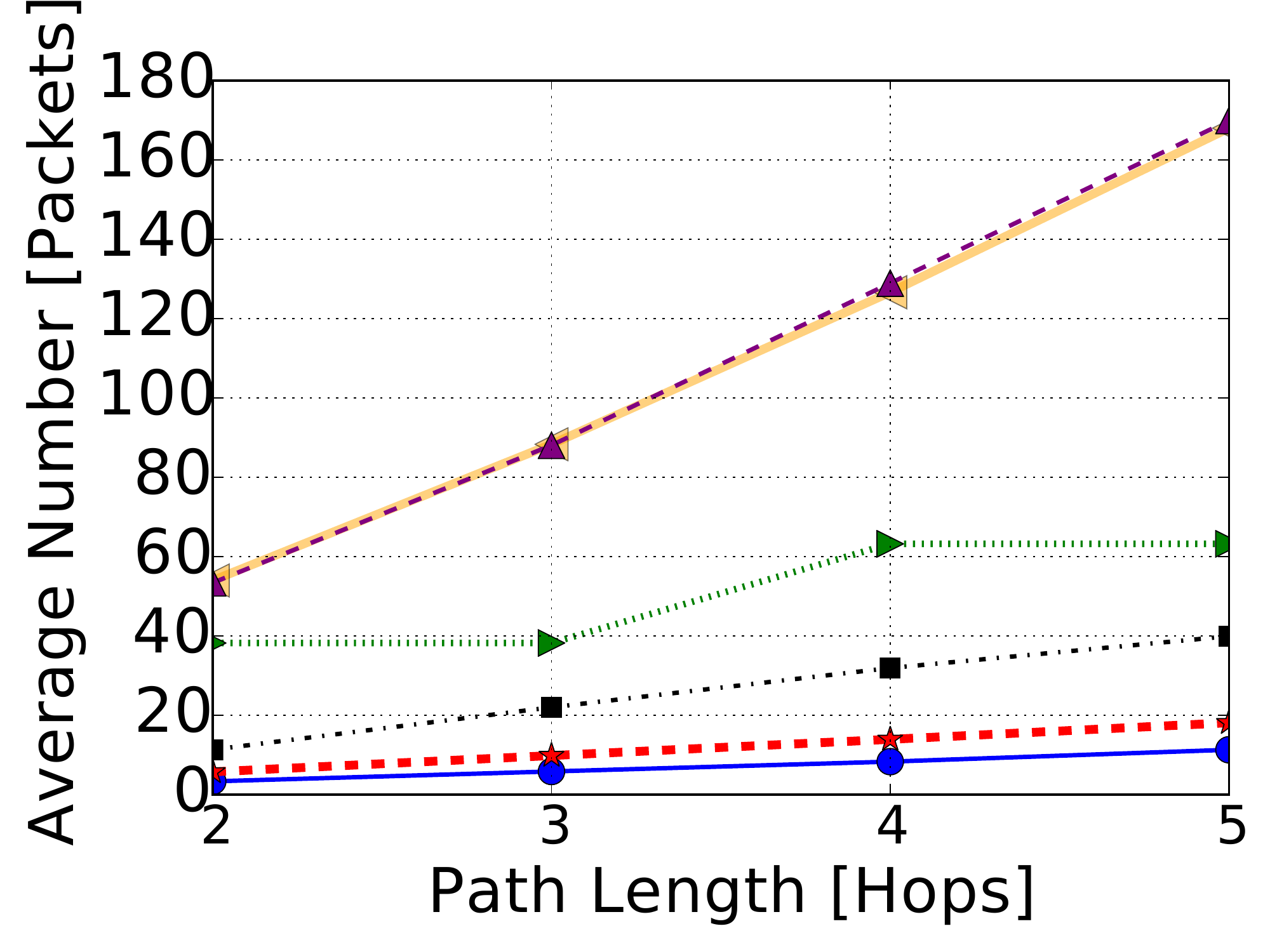}}\\\hspace*{19mm}
\label{fig:path_legend}\includegraphics[width=0.58\linewidth]{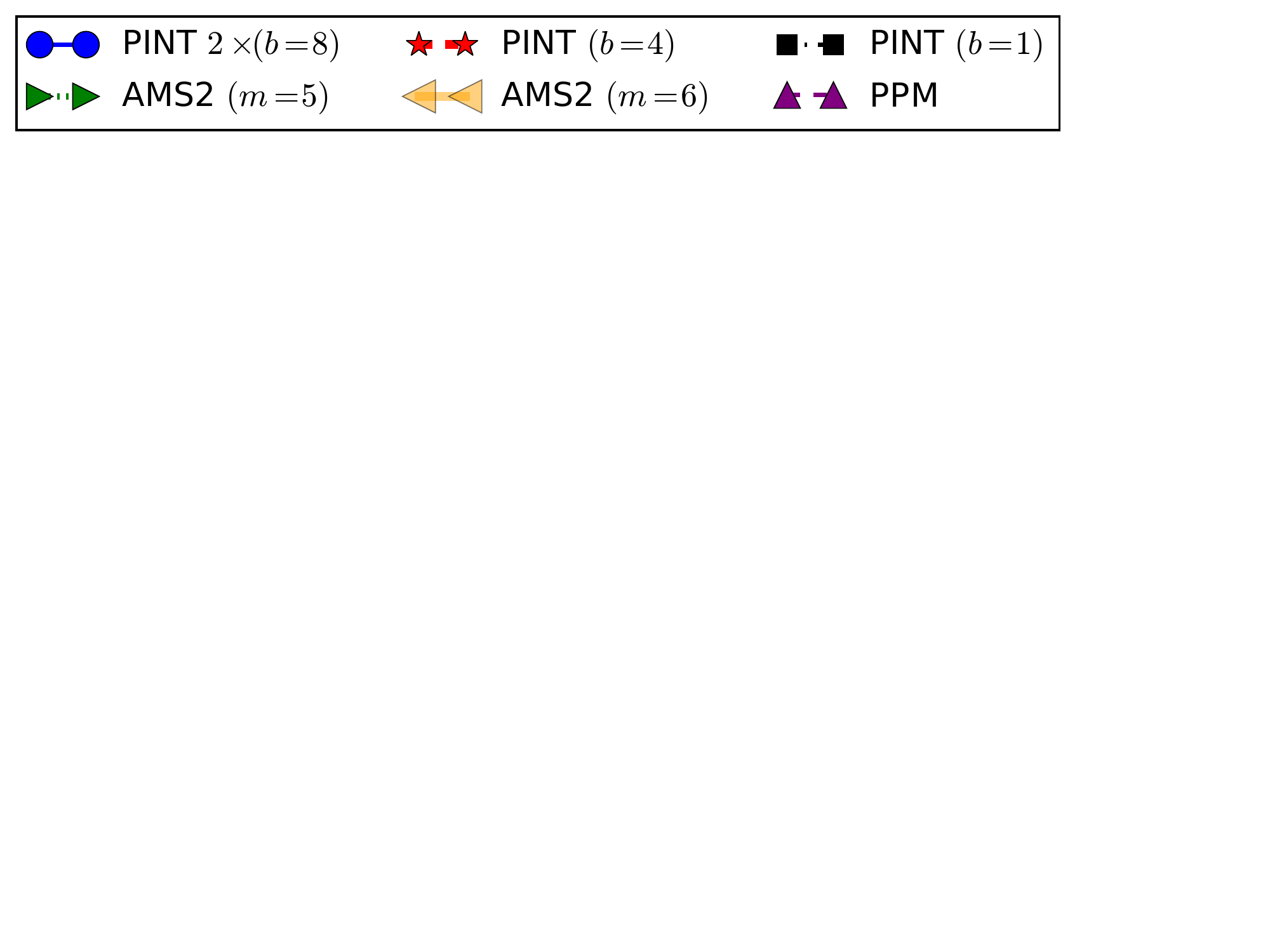}\\
\vspace*{-68mm} 
\subfigure[Kentucky Datalink ($D=59$)]{\label{fig:d59_tail}\includegraphics[width=.33\linewidth]{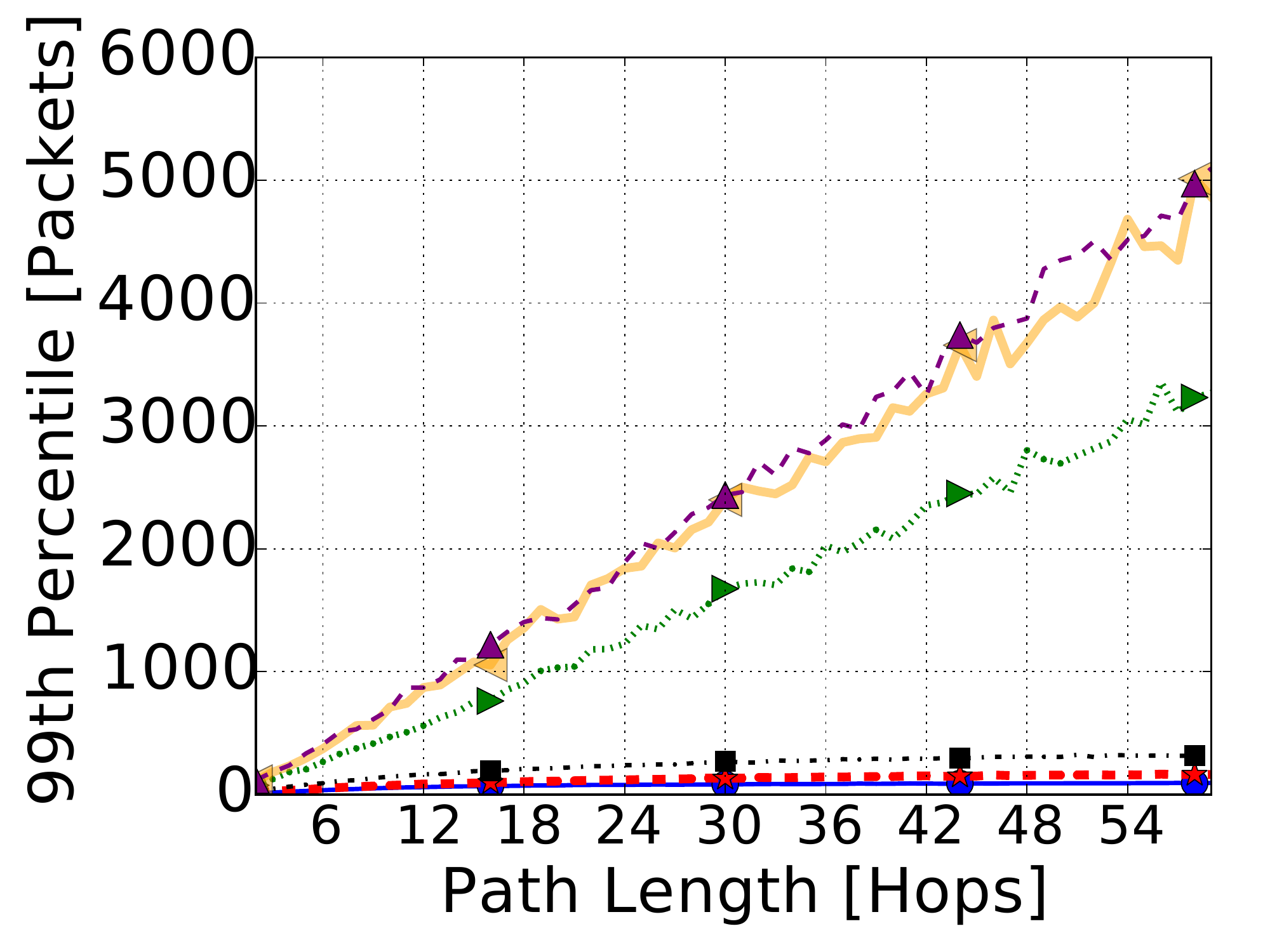}}
\subfigure[US Carrier ($D=36$)]{\label{fig:d36_tail}\includegraphics[width=.33\linewidth]{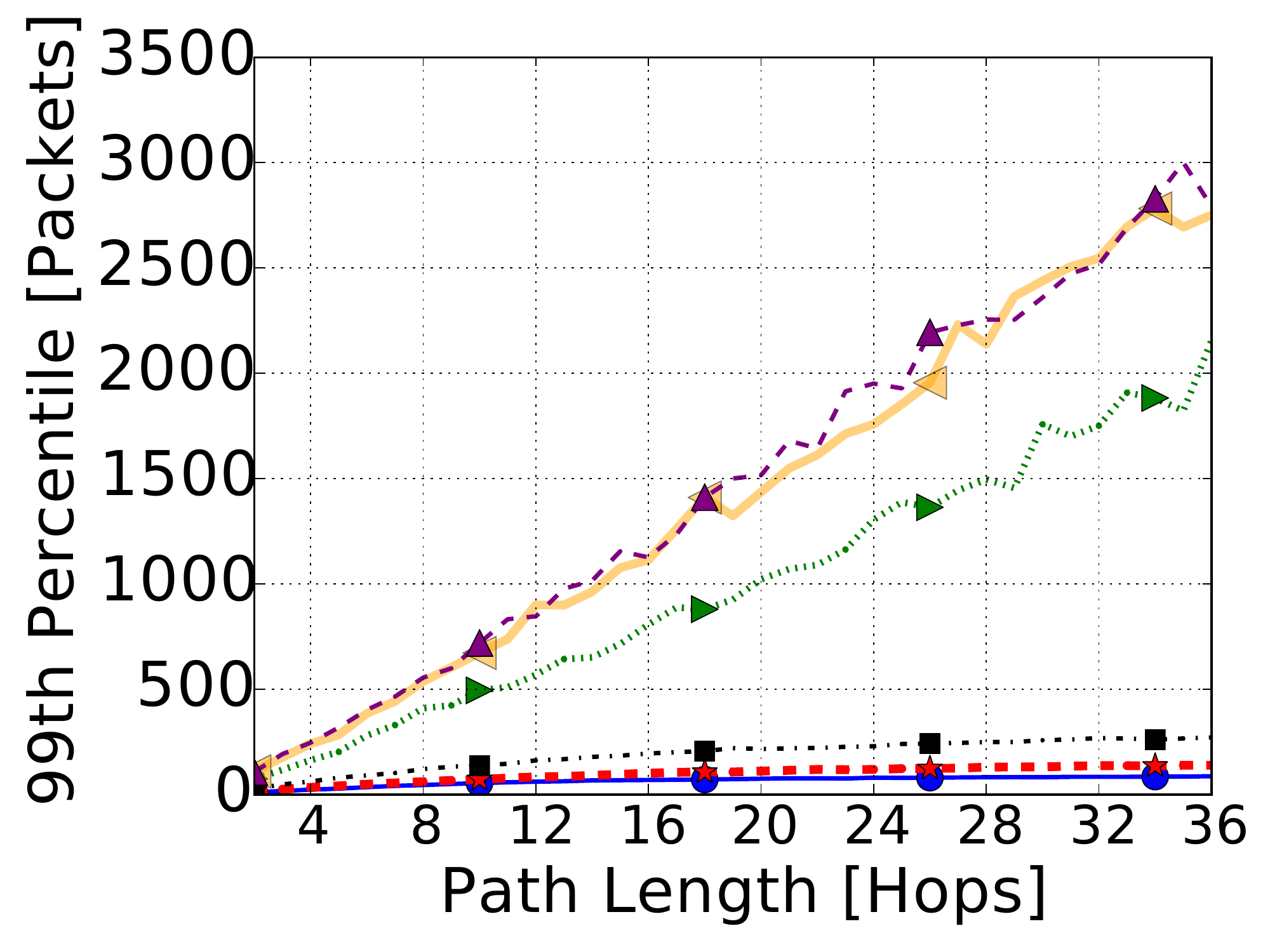}}
\subfigure[Fat Tree ($D=5$)]{\label{fig:ft_tail}\includegraphics[width=.33\linewidth]{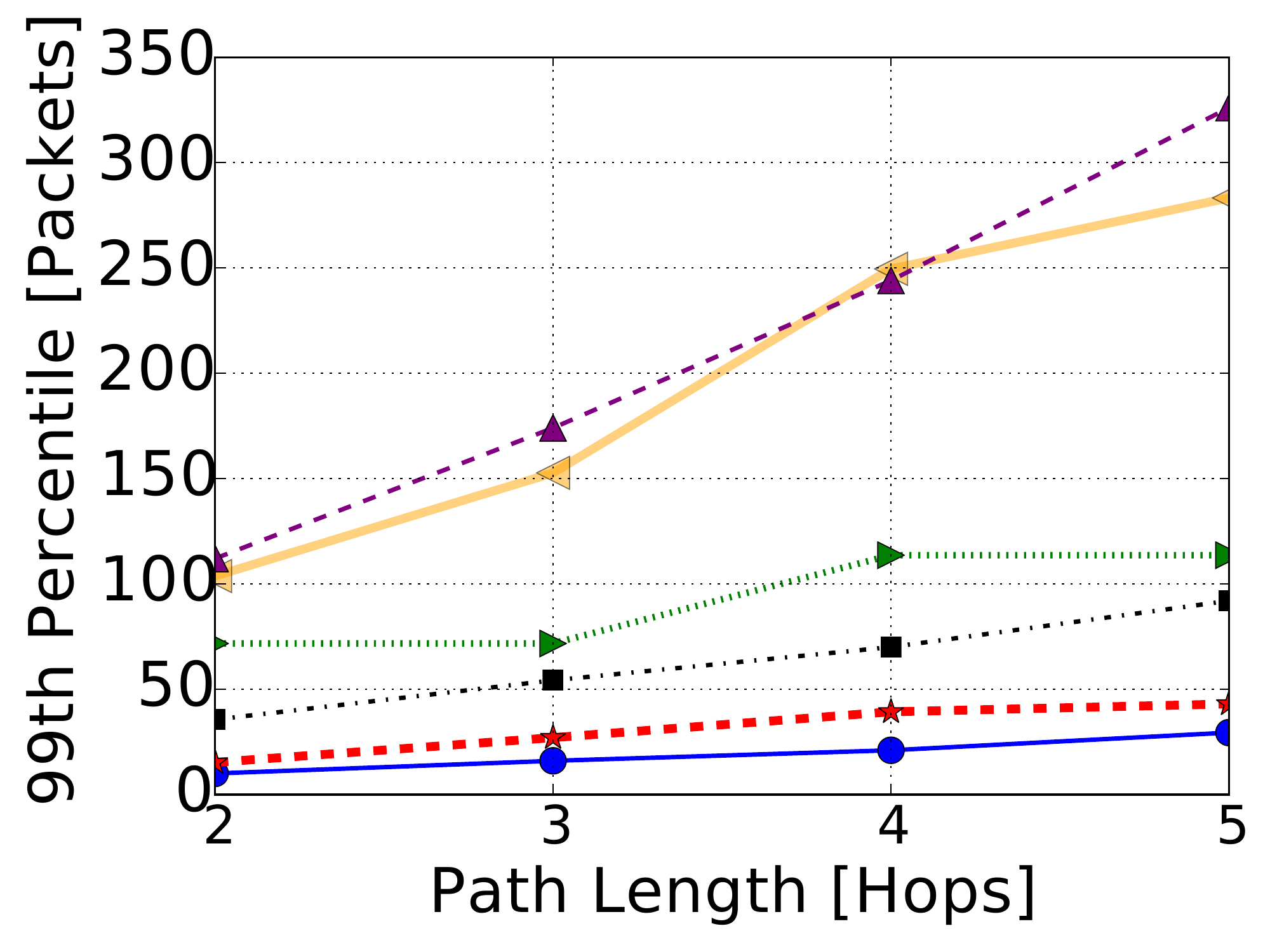}}

\caption{\small Comparison of the number of packets required (lower is better) for path decoding of different algorithms, including \sys \mbox{with varying bit-budget.}}
\label{fig:pathEval}

\end{figure*}

\subsection{Latency Measurements}
\changed{\sout{We set up a Clos topology in Mininet with 4 spine switches, 16 leaf switches, and 20 hosts in each leaf switch. We run the Facebook traces~\cite{roy15} on this topology and measure the accuracy of estimating the median and tail latencies for each (switch, flow) pair.}}

\changed{Using the same topology and workloads as in our congestion control experiments, we evaluate \sys's performance on estimating latency quantiles.}
We consider \sys in \changed{four} scenarios,  using $b=4$ and $b=8$ bit-budgets\changed{, with sketches (denoted \sys{}$_S$), and without}. 
\changed{
In our experiment, we have used the, state of the art, KLL sketch~\cite{KLL16}.}
The results, appearing in~\cref{fig:latency}, show that when getting enough packets, the error of the aggregation becomes stable and converges to the error arising from compressing the values.
\changed{\sout{For obtaining tail and median latency within 10\% error, \sys with a $b=8$ bit-budget requires roughly 25 and 100 packets, respectively.}}
\changed{As shown, by compressing the incoming samples using a sketch (e.g., that keeps $100$ identifiers regardless of the number of samples), \sys accuracy degrades only a little even for small $100$B sketches. We conclude that such sketches offer an attractive space \mbox{to accuracy tradeoff.} 
}

\subsection{Path Tracing}
We conduct these experiments on Mininet~\cite{MininetA59:online} using two large-diameter (denoted $D$) ISP topologies (Kentucky Datalink and US Carrier) from Topology Zoo~\cite{knight2011internet} and a ($K=8$) Fat Tree topology. \changed{The Kentucky Datalink topology consisted of 753 switches with a diameter of 59 and the US carrier topology consisted of 157 switches with a diameter of 36.} 
For each topology and every path, we estimate the average and 99'th percentile number of packets needed for decoding over 10K runs. We consider three variants of \sys -- using 1-bit, 4-bit, and two independent 8-bit hash functions (denoted by $2\times(b=8)$). We compare \sys to two state-of-the-art IP Traceback solutions PPM~\cite{savage00} and AMS2~\cite{song2001advanced} with $m=5$ and $m=6$. When configured with $m=6$, AMS2 requires more packets to infer the path but also has a lower chance of false positives (multiple possible paths) compared with $m=5$.
We implement an improved version of both algorithms using Reservoir Sampling, \mbox{as proposed in~\cite{sattari2007revisiting}.}
\sys is configured with $d=10$ on the ISP topologies and $d=5$ (as this is the diameter) on the fat tree topology. In both cases, this means a single XOR layer
in addition to a Baseline layer.

The results (\cref{fig:pathEval}) show that \sys significantly outperforms previous works, even with a bit-budget of a single bit (PPM and AMS both have an overhead of $16$ bits per packet). As shown, the required number of packets for \sys grows near-linearly with the path length, validating our theoretical analysis. For the Kentucky Datalink topology ($D=59$), \sys with $2 \times (b=8)$ on average uses 25--36 times fewer packets when compared to competing approaches. Even when using \sys with $b=1$, \sys needs 7--10 times fewer packets than competing approaches.
For the largest number of hops we evaluated ($59$, in the Kentucky Datalink topology), \sys requires only 42 packets on average and 94 for the 99'th percentile, while alternative approaches need at least 1--1.5K on average and 3.3--5K for 99'th percentile, respectively.

\begin{figure}[b]
        \centering
        \includegraphics[width=\linewidth]{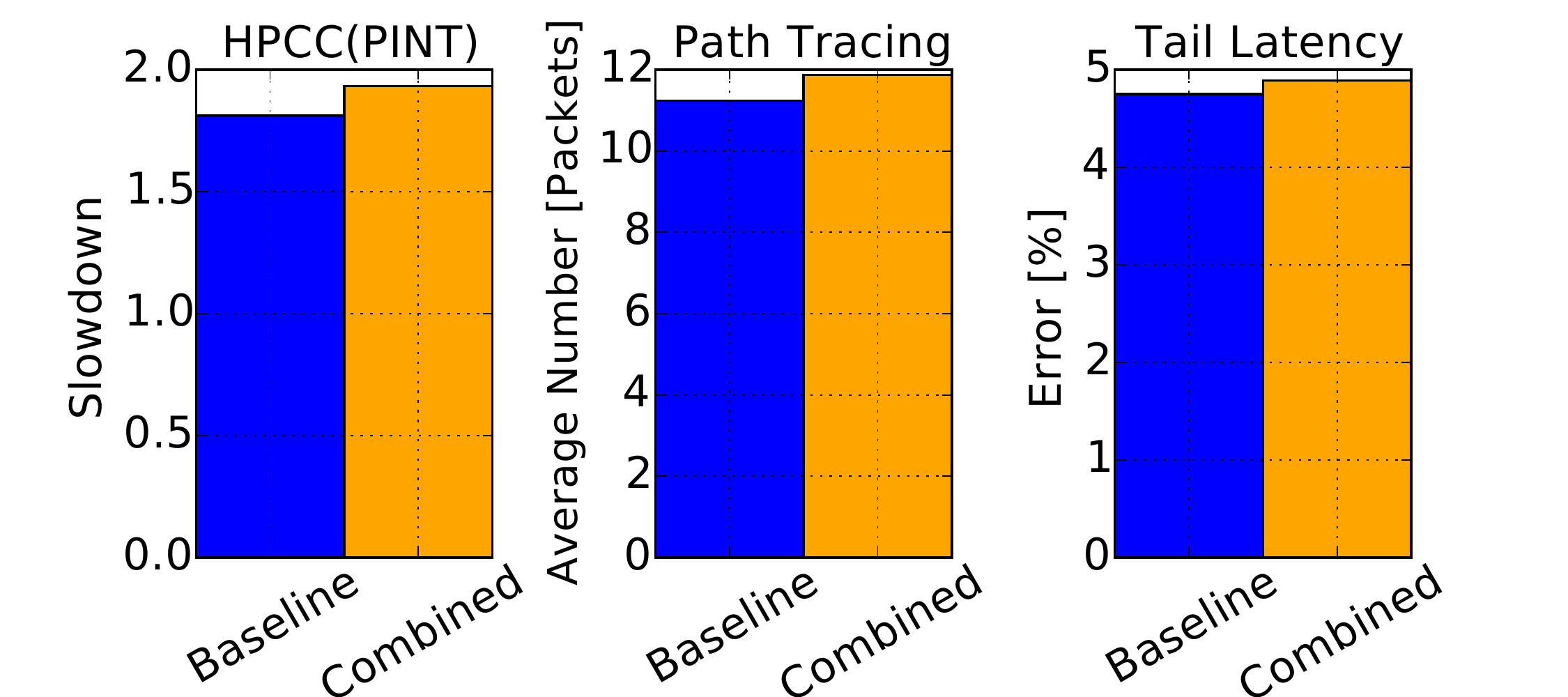}

        \caption{\small
        {The performance of each query in a concurrent execution (FatTree topology + Hadoop workload) compared to running it alone.}
        }

        \label{fig:Combined}
        
\end{figure}

\subsection{Combined Experiment}
We test the performance of \sys when running all three use cases concurrently. Based on the previous experiments, we tune \sys to run each query using a bit budget of $8$ bits and a global budget of $16$ bits.
\changed{Our goal is to compare how \sys performs in such setting, compared with running each application alone using $16$ bits per packet (i.e., with an effective budget of $3\times 16$ bits).}
That is, each packet can carry digests of two of the three concurrent queries.  As we observe that the congestion control application has good performance when running in $p=1/16$ of the packets, and the path tracing requires more packets than the latency estimation, we choose the following configuration.
We run the path algorithm on all packets, alongside the latency algorithm in $15/16$ of the packets, and alongside HPCC in $1/16$ of the packets.  As~\cref{fig:Combined} shows, the performance of \sys is close to a Baseline of running each query separately. For estimating median latency, the relative error increases by only 0.7\% from the Baseline to the combined case. In case of HPCC, we that observe short flows become 6.6\% slower while the performance of long flows does not degrade. As for path tracing, the number of packets increases by 0.5\% compared with using two $8$ bit hashes as in \cref{fig:pathEval}. We conclude that, with a tailored execution plan, our system can support these multiple concurrent telemetry queries using an \mbox{overhead of just two bytes per packet.}

\changed{\section{Limitations}
In this section, we discuss the limitations associated with our
probabilistic approach. The main aspect to take into consideration
is the required per-packet bit-budget and the network diameter.
The bigger overhead allowed and the smaller the network,
the more resilient \sys will be in providing results in different
scenarios.}

\noindent\changed{\textbf{Tracing short flows.}
\sys leverages multiple packets from the same flow to infer its path.
In our evaluation (\S\ref{sec:eval}), we show that our solution needs significantly fewer packets when compared to competing approaches.
However, in data center networks, small flows can consist of
just a single packet~\cite{alizadeh10}. In this case, \sys is not
effective and a different solution, such as INT, \mbox{would provide the
required information.}}

\noindent\changed{\textbf{Data plane complexity.}
Today's programmable switches have a limited number of pipeline stages.
Although we show that it is possible to parallelize the processing
of independent queries (\S\ref{sec:impl}), thus saving resources,
the \sys requirements might restrict the amount of additional use cases
to be implemented in the data plane, e.g., fast reroute~\cite{chiesa19} or 
in-network caching~\cite{jin17} and load balancing~\cite{alizadeh14,katta16}.}

\noindent\changed{\textbf{Tracing flows with multipath routing.}
The routing of a flow may change over time (e.g., when using flowlet
load balancing~\cite{alizadeh14,katta16}) or multiple paths can be
taken simultaneously when appropriate transport protocols such
as Multipath TCP are used~\cite{rfc6824}. In those cases, the values 
(i.e, switch IDs) for some hops will be different. Here, \sys can 
detect routing changes when observing a digest that is not consistent 
with the part of the path inferred so far. For example, if we know that 
the sixth switch on a path is $M_6$, and a Baseline packet $p_j$ comes 
with a digest from this hop that is different than $h(M_6,p_j)$, then 
we can conclude that the path has changed. The number of packets needed 
to identify a path change depends on the fraction of the path that has 
been discovered. If path changes are infrequent, and \sys knows the 
entire path before the change, a Baseline packet will not be consistent 
with the known path (and thus signify a path change) with probability 
$1-2^{-\mathfrak q}$. Overall, in the presence of flowlet routing, \sys can still 
trace the path of each flowlet, provided enough packets for 
each flowlet-path are received at the sink. \sys can also 
profile all paths simultaneously at the cost of additional overhead 
(e.g., by adding a path checksum to packets we can associate each with 
the path it followed).}

\noindent\changed{\textbf{Current implementation.} At the time of writing, 
the \sys execution plan is manually selected. We envision that an end to 
end system that implements \sys would include a Query Engine that automatically 
decides how to split the bit budget. }

\section{Related Work}
Many previous works aim at improving data plane visibility.
Some focus on specific flows selected by operators~\cite{narayana16,tilmans18,zhu15} 
or only on randomly selected sampled flows~\cite{duffield01,AROMA}. Such approaches are insufficient for applications that need global visibility on all flows, such as path tracing. Furthermore, the flows of interest may not be known in advance, if we wish to debug high-latency \mbox{or malicious flows.}

Other works can be classified into three main approaches: (1) keep information out-of-band; (2) keep flow state at switches; or (3) keep information on packets. The first approach applies when the data plane status is recovered by using packet mirroring at switches or by employing specially-crafted probe packets.
Mirroring every packet creates scalability concerns for both trace collection and analysis. The traffic in a large-scale data center network with hundreds of thousands of servers can quickly introduce terabits of mirrored traffic~\cite{roy15,guo15}.
Assuming a CPU core can process tracing traffic at 10 Gbps, thousands of cores would be required for trace analysis~\cite{zhu15}, which is prohibitively expensive. Moreover, with mirroring it is not possible to retrieve information related to switch status, such as port utilization or queue occupancy, that are of paramount importance for applications such as congestion control or 
network troubleshooting. While such information can be retrieved with specially-crafted probes~\cite{tan2019netbouncer}, the feedback loop may be too slow for applications like high precision congestion control~\cite{li19}. 
We can also store flow information at switches and periodically export it to a collector~\cite{snoeren01,li16}. However, keeping state for a large number of active 
flows (e.g., up to 100K~\cite{roy15}), in the case of path tracing, is challenging for limited switch space (e.g., 100 MB~\cite{miao17}). This is because operators need the memory for essential control functions such as ACL rules, customized forwarding~\cite{sivaraman15}, and other network functions and applications~\cite{miao17, jin17}.
Another challenge is that we may need to export data plane status frequently (e.g., every 10 ms) to the collector, if we want to enable applications such as congestion control. This creates significant \mbox{bandwidth and processing overheads~\cite{li16}.}

Proposals that keep information on packets closely relate to this work~\cite{int,jeyakumar14,tammana16}, with INT being considered the state-of-the-art solution. 
\changed{Some of the approaches, e.g., Path Dump~\cite{tammana16}, show how to leverage properties of the topology to encode only part of each path (e.g., every other link). Nonetheless, this still imposes an overhead that is linear in the path length, while \sys keeps it constant.}
\changed{Alternative approaches add small digests to packets for tracing paths~\cite{sattari2010network,savage00,song2001advanced}. However, they attempt to trace back to potential attackers (e.g., they do not assume unique packet IDs or reliable TTL values as these can be forged) and require significantly more packets for identification, as we show in~\cref{sec:eval}}.
In a recent effort to reduce overheads
on packets, similarly to this work, Taffet et al.~\cite{taffet19} propose having switches use Reservoir Sampling to collect information about a packet's path and congestion that the packet encounters as it passes through the network. \sys takes the process several steps further, including approximations and coding (XOR-based or network coding) to reduce the cost of adding information to packets as much as 
possible. Additionally, our work rigorously proves performance bounds on the number of packets required to recover the data plane status as well as proposes trade-offs between data size and time to recover.

\section{Conclusion}
We have presented \sys, a probabilistic framework to in-band telemetry that provides similar visibility to INT while bounding the per-packet overhead to a user-specified value. 
This is important because overheads imposed on packets translate to inferior flow completion time and application-level goodput. 
We have proven performance bounds (deferred to \cref{app:analysis} due to lack of space) for \sys and have implemented it in P4 to ensure it can be readily deployed on commodity switches.
\changed{\sys goes beyond optimizing INT by removing the header and using succinct switch IDs by restricting the bit-overhead to a constant that is independent of the path length.}
We have discussed the generality of \sys and demonstrated its performance on three specific use cases: path tracing, data plane telemetry for congestion control and estimation of experienced median/tail latency. Using real topologies and traffic characteristics, we have shown that \sys enables the use cases, while drastically decreasing the required overheads on packets with respect to INT.

\vspace{0.04in} 
\changed{\parab{Acknowledgements.}
We thank the anonymous reviewers, Jiaqi Gao, Muhammad Tirmazi, and our shepherd, Rachit Agarwal, for helpful comments and feedback. 
This work is partially sponsored by EPSRC project EP/P025374/1, by NSF grants \#1829349, \#1563710, and \#1535795, and \mbox{by the Zuckerman Foundation.} \\
\ \vspace*{0.08in}\hspace*{-1mm}
}

%% file: Appendix.tex
\section{Analysis}\label{app:analysis}
\subsection{Dynamic per-flow Aggregation}\label{app:dynamic}
We now survey known results that \cref{thm:dynamicQuantiles} and \cref{thm:dynamicFrequent} are based on.

\parab{Quantiles.}
Classical streaming results show that by analyzing a uniformly selected subset of $O(\epsilon_s^{-2}\log\epsilon_s^{-1})$ elements, one can estimate all possible quantiles~\cite{manku1998approximate,vapnik2015uniform} to within an additive error of $\epsilon_s$. It is also known that if one is interested in a specific quantile (e.g., median), a subset size \mbox{of $O(\epsilon_s^{-2})$ is enough.}

In our case, for each switch $s_i$ through which flow $x$ is routed, we get a sampled substream of $S_{i,x}$ where each packet carries a value from it with probability $1/k$. This is not a fixed-size subset, but a \emph{Bernouli sample}. Nevertheless, Felber and Ostrovsky show that a Bernouli sample with the same \emph{expected size} is enough~\cite{felber2017randomized}.
Therefore, for a specific quantile (e.g., median) we need to get $O(\epsilon_s^{-2})$ samples for each of the $k$ switches on the path. 
Using a standard Chernoff bound argument, we have that if $z=O(k\epsilon_s^{-2})$ packets reach the \sys sink, \emph{all} hops on the path will get at least $O(\epsilon_s^{-2})$ samples with probability $1-e^{-\Omega(z/k)} = 1-e^{-\Omega(\epsilon_s^{-2})}$. 

To compress the amount of per-flow storage needed for computing the quantiles, we can use a $\widetilde O(\epsilon_a^{-1})$ space sketch such as KLL~\cite{KLL16}.
We run separate sketch for each of the $k$ hops, thus needing $\widetilde O(k\epsilon_a^{-1})$ per-flow storage in total. The resulting error would be $\epsilon=\epsilon_s+\epsilon_a$, as the sampling adds an additive error of $\epsilon_s$ and the sketching an additive error of $\epsilon_a$.

\parab{Frequent Values.}
Using a standard Chernoff bound argument, one can use a $O(\epsilon_s^{-2})$-sized substream of $S_{i,x}$, one can estimate the fraction in which each specific value appears, up to an additive error of $\epsilon_s$. We can then use a heavy hitters algorithm like Space Saving~\cite{SpaceSavings} to estimate the frequency of values in the sampled substream to within an additive error of $\epsilon_a$, using $O(\epsilon_a^{-1})$ space. 
As before, to get the correct estimations for all hops, we need a factor $k$ multiplicative overhead to both the number of packets and space.

\subsection{Static per-flow Aggregation}\label{app:static}
Before we can analyze the algorithm (\S\ref{app:actualStaticApp}), we start with some auxiliary results.
\subsubsection{Auxiliary Results}
The first lemma gives a bound on how many independent coins with probability $p$ we need to flip until we get $k$ successes.
\begin{lemma}\label{lem:negativeBinomialBound}
Let $k\in\mathbb N$ and $p,\delta\in (0,1)$. \\Denote $N=\frac{k+2\ln\delta^{-1}+\sqrt{2k\ln\delta^{-1}}}{p}$ and let $X\sim\mathit{Bin}(N,p)$. Then
$$
\Pr[X \le k]\le \delta.
$$
\end{lemma}
\begin{proof}
Using the Chernoff bound we have that for any $\gamma>0$:
$$
\Pr[X<\mathbb E[X](1-\gamma)]\le e^{-\gamma^2\mathbb E[X]/2}\ .
$$
We set $\gamma=\sqrt{\frac{2\ln\delta^{-1}}{Np}}$, which means that $\gamma^2\mathbb E[X]/2=\ln\delta$ and therefore $\Pr[X<\mathbb E[X](1-\gamma)]\le\delta$.

Finally,
\begin{multline*}
\mathbb E[X](1-\gamma) = Np(1-\gamma) = 
Np - \sqrt{2{Np}\ln\delta^{-1}}.
\end{multline*}
Denote $x=\sqrt{Np}$, then we want to show that
$$
x^2 - x\sqrt{2\ln\delta^{-1}}-k\ge 0,
$$
which holds for 
$$
x>\frac{\sqrt{2\ln\delta^{-1}}+\sqrt{2\ln\delta^{-1}+4k}}{2}
= \frac{\sqrt{\ln\delta^{-1}}+\sqrt{\ln\delta^{-1}+2k}}{\sqrt 2}.
$$
This gives
\begin{multline*}
N = x^2/p \ge \frac{\parentheses{\sqrt{\ln\delta^{-1}}+\sqrt{\ln\delta^{-1}+2k}}^2}{2p}\\
= \frac{k+2\ln\delta^{-1}+\sqrt{2k\ln\delta^{-1}}}{p}.
\end{multline*}
\end{proof}

The next theorem provide a high-probability bound on the Double Dixie Cup problem~\cite{DoubleDixie}. Specifically, consider trying to collect at least $Z$ copies from each of $k$ coupons, where at each stage you get a random coupon. The following bounds the number of samples you need.
\begin{theorem}\label{thm:NVal}
After seeing {\small $$N=\NVal$$}
samples, our algorithm has at least $Z$ copies of each of the $k$ coupons.
\end{theorem}
\begin{proof}
Therefore, the number of copies of the $i$'th coupon isa binomial random variable we denote by $Y_i\sim \mathit{Bin}(N,1/k)$.
Our goal is to show that getting $Y_i< Z$ is unlikely; to that end, we use the Chernoff inequality that states that $\Pr[Y_i\le \mathbb E[Y_i](1-\gamma)]\le e^{-\mathbb E[Y_i]\gamma^2/2}$ for any $\gamma\in(0,1]$. We set $\gamma=1-k(Z-1)/N$ to get
\begin{multline*}
    \Pr[Y_i< Z] = \Pr[Y_i\le Z-1] = \Pr[Y_i\le \mathbb E[Y_i](1-\gamma)]
    \\ \le e^{-\mathbb E[Y_i]\gamma^2/2} 
    = e^{-N/2k\cdot (1-k(Z-1)/N)^2}
    \\= e^{-(N/2k -(Z-1)+k(Z-1)^2/2N)}
    = e^{-(x -(Z-1)+(Z-1)^2/4x)},
\end{multline*}
where $x=N/(2k)$. We want $\Pr[X_i< Z]\le \delta/k$, which according to the above follows from
\begin{multline*}
    x -(Z-1)+(Z-1)^2/4x\ge \ln k/\delta \iff\\
     \mbox{\small $x\ge \xVal$.}
\end{multline*}
The last inequality follows directly from our choice of $N$. Finally, we use the union bound to get that after $N$ samples all coupons get at least $Z$ copies except with probability $k\cdot \Pr[Y_i< Z]\le \delta$.
\end{proof}
We proceed with a tail bound on the Partial Coupon Collector problem, in which we wish to get $N$ out of $r$ possible coupons, where at each timestamp we get a random coupon.
Our proofs relies on the following result for a sharp bound on the sum of geometric random variables:
\begin{theorem}\label{thm:geoSum}(\cite{janson2018tail})
Let $\{A_1,\ldots A_N\}$ be independent geometric random variables such that $A_i\sim Geo(p_i)$ and $p_1\ge\ldots\ge p_N$. Then the sum $A=\sum_{i=1}^N A_i$ satisfies:
$$
\Pr\brackets{A > \mathbb E[A]\cdot \lambda}\le e^{-p_N\mathbb E[A](\lambda-1-\ln\lambda)}.
$$
\end{theorem}
Additionally, we will use the following fact.
\begin{fact}\label{fact}
For any positive real number
$\epsilon\in\mathbb R^+$, 
$$1+\epsilon+\sqrt{2\epsilon}-\ln(1+\epsilon+\sqrt{2\epsilon})\ge 1+\epsilon.$$
\end{fact}
We now prove our result.
\begin{theorem}\label{thm:partialCouponCollector}
Let $\mathbb E[A]= r(H_r-H_{r-N})$ denote the \emph{expected} number of samples required for seeing $N$ distinct coupons.
With probability $1-\delta$, the number of samples required for seeing at least $N$ distinct coupons is at most
$$
{\mathbb E[A]+\frac{r\ln\delta^{-1}}{(r-N)}+\sqrt{\frac{2r\mathbb E[A]\ln\delta^{-1}}{(r-N)}}}.
$$
\end{theorem}
\begin{proof}
We wish to use Theorem~\ref{thm:geoSum}; notice that 
we need $\lambda-\ln\lambda\ge 1+\frac{\ln\delta^{-1}}{p_N\mathbb E[A]}$ which implies
$$
e^{-p_N\mathbb E[A](\lambda-1-\ln\lambda)}\le \delta.
$$
According to Fact~\ref{fact}, for $\epsilon=\frac{\ln\delta^{-1}}{p_N\mathbb E[A]}$, it is enough to set 
$$\lambda = 1+\frac{\ln\delta^{-1}}{p_N\mathbb E[A]}+\sqrt{\frac{2\ln\delta^{-1}}{p_N\mathbb E[A]}}.$$
Plugging in $p_N=(1-(N-1)/r) > (r-N)/r$
we have that the required number of required packets is at most
\begin{multline*}
\lambda\cdot \mathbb E[A] = \parentheses{\mathbb E[A]+\frac{\ln\delta^{-1}}{p_N}+\sqrt{\frac{2\mathbb E[A]\ln\delta^{-1}}{p_N}}}\\
\le \parentheses{\mathbb E[A]+\frac{r\ln\delta^{-1}}{(r-N)}+\sqrt{\frac{2r\mathbb E[A]\ln\delta^{-1}}{(r-N)}}}
.
\end{multline*}
For example, if $r=2N$, we have $\mathbb E[A]\approx 1.39 N$ and the number of packets is 
\begin{multline*}
\parentheses{\mathbb E[A]+{2\ln\delta^{-1}}+\sqrt{{4\mathbb E[A]\ln\delta^{-1}}}}\\\approx
\parentheses{1.39 N+{2\ln\delta^{-1}}+2.35\sqrt{N\ln\delta^{-1}}}
.\qedhere
\end{multline*}
\end{proof}
Next, we show a bound on the number of samples needed to collect $\mathcal K(1-\psi)$ in a Coupon Collector process~\cite{flajolet1992birthday} on $\mathcal K$~coupons.
\begin{lemma}\label{lem:partialCoverage}
Let $\mathcal K\in\mathbb N^+$ and $\psi\in(0,1/2]$.
The number of samples required for collecting all but $\psi\mathcal K$ coupons is at most
\begin{multline*}
\mathcal K\ln\psi^{-1}+\psi^{-1}\ln\delta^{-1} + \sqrt{2\mathcal K\psi^{-1}\ln\psi^{-1}\ln\delta^{-1}} \\
= O(\mathcal K\ln\psi^{-1}+\psi^{-1}\ln\delta^{-1}).
\end{multline*}
\end{lemma}
\begin{proof}
For $i=1,\ldots,\mathcal K(1-\psi)$, let $A_i\sim\mathit{Geo}(1-(i-1)/\mathcal K)$ denote the number of samples we need for getting the $i$'th distinct coupon, and let $A=\sum_{i=1}^{K(1-\psi)}A_i$.
We have that 
$$
\mathbb E[A] = \sum_{i=1}^{\mathcal K(1-\psi)}\frac{\mathcal K}{\mathcal K-(i-1)} = \mathcal K\parentheses{H_{\mathcal K}-H_{\mathcal K\psi}} = \mathcal K\ln\psi^{-1}.
$$

According to Theorem~\ref{thm:partialCouponCollector}, it is enough to obtain the following number of samples
\begin{multline*}
{\mathbb E[A]+\frac{\mathcal K\ln\delta^{-1}}{\mathcal K(1-(1-\psi))}+\sqrt{\frac{2\mathcal K\mathbb E[A]\ln\delta^{-1}}{\mathcal K(1-(1-\psi))}}}\\
= \mathcal K\ln\psi^{-1}+\psi^{-1}\ln\delta^{-1} + \sqrt{2\mathcal K\psi^{-1}\ln\psi^{-1}\ln\delta^{-1}}.
\end{multline*}
Finally, we note that $\sqrt{\mathcal K\psi^{-1}\ln\psi^{-1}\ln\delta^{-1}}$ is the geometric mean of $\mathcal K\ln\psi^{-1}$ and $\psi^{-1}\ln\delta^{-1}$ and thus:
\begin{multline*}
\mathcal K\ln\psi^{-1}+\psi^{-1}\ln\delta^{-1} + \sqrt{2\mathcal K\psi^{-1}\ln\psi^{-1}\ln\delta^{-1}}\\\le
\parentheses{\mathcal K\ln\psi^{-1}+\psi^{-1}\ln\delta^{-1}}(1+1/\sqrt 2)\\=
O\parentheses{\mathcal K\ln\psi^{-1}+\psi^{-1}\ln\delta^{-1}}.
\end{multline*}
\end{proof}
\subsubsection{Analysis of the algorithm}\label{app:actualStaticApp}
We denote by $\mathfrak d\triangleq \frac{d}{\log^* d}$ the number of hops we aim to decode using the XOR layers.
Our algorithm has $\ceil{\log^* \mathfrak d} + 1$ \emph{layers}, where layer $0$ runs the Baseline scheme and the remaining $\mathcal L\triangleq \ceil{\log^* \mathfrak d}$ layers use XOR. 
We denote by $\duparrow$ Knuth's iterated exponentiation arrow notation, \mbox{i.e., $x\duparrow 0 = 1$ and}
$$ \left.\kern-\nulldelimiterspace
  \begin{array}{@{}c@{}}
    x\duparrow y = x^{x^{\scriptstyle x^{\cdot^{\cdot^{\cdot^{\scriptstyle x}}}}}}
  \end{array}
  \right\rbrace
  \text{\scriptsize $y$-times}.
  $$
The sampling probability in layer $\ell$ is then set to 
$$p_\ell =  \frac{{e\duparrow (\ell-1)}}{\mathfrak{d}}.$$ 
Each packet it hashed to choose a layer, such that layer $0$ is chosen with probability $\tau=\parentheses{1-\frac{1}{1+\log \log^* \mathfrak d}}=1-o(1)$ and otherwise one of layers $1,\ldots,\mathcal L$ is chosen uniformly.
The pseudo code for the final solution is given in Algorithm~\ref{alg:final}.
\begin{algorithm}[H]
\algSize
\caption{\sys Processing Procedure at Switch $s$}
\label{alg:final}
\begin{algorithmic}[1]
    \Statex\textbf{Input:} A packet $p_j$ with $\mathfrak b$-bits digest $p_j.\mbox{dig}$.
    \Statex\textbf{Output:} Updated digest $p_j.\mbox{dig}$.
    \Statex\textbf{Initialization:}
    \Statex{$\tau=\frac{\log \log^* \mathfrak d}{1+\log \log^*\mathfrak d}$, $\forall \ell\in\set{1,\ldots,\mathcal L}: p_\ell = \frac{e\duparrow (\ell-1)}{\mathfrak{d}}$.}
    \Statex\hrulefill
    \Statex Let $i$ such that the current switch is the $i'th$ so far
    \State $\mathfrak H\gets \mathcal H(p_j)$ \Comment{Distributed uniformly on $[0,1]$}
    \If{$\mathfrak H < \tau$}\Comment{Update layer $0$}
        \If{$g(p_j,i)$ < $1/i$}
                \State $p_j.\mbox{dig} \gets h(s,p_j)$\label{line:sampleOutgoing25}\Comment{Sample with probability $1/i$}
        \EndIf
    \Else
        \State $\ell \gets \ceil{\mathcal L\cdot\frac{\mathfrak H-\tau}{1-\tau}}$\Comment{Choose the layer}
        \If{$g(p_j,i)$ < $p_\ell$}
                \State $p_j.\mbox{dig} \gets p_j.\mbox{dig}\oplus h(s,p_j)$\label{line:sampleOutgoing2}\Comment{Xor w.p. $p_\ell$}
        \EndIf
    \EndIf
\end{algorithmic}
\end{algorithm}
For simplicity, we hereafter assume in the analysis that a packet can encode an entire identifier. This assumption is not required in practice and only serves for the purpose of the analysis. We note that even under this assumption the existing approaches require $O(k\log k)$ packets.
In contrast, we show that except with probability $\delta= e^{-O(k^{0.99})}$ the number of packets required for decoding a $k$-hops path in our algorithm is just 
$$\mathcal X = k\log\log^* k\cdot (1+o(1)).\footnotemark$$
\footnotetext{The $o(1)$ part hides an additive $O(k)$ term, which we upper bound as $\frac{k}{c\cdot e^{-c}}$ up to lower order terms. Specifically, if $d = k$ (thus, $c=1$), the required number of packets is at most $k\log\log^* k + e\cdot k + o(k)$.}
Note that $\log\log^* k$ is a function that grows extremely slowly, e.g., $\log\log^* P<2$ where $P$ is the number of atoms in the universe. 
Our assumption on the error probability $\delta$ allows us to simplify the expressions and analysis but we can also show an $$O\parentheses{k\log\log^* k + \log^* k\log\delta^{-1}}$$ bound on the required number of packets thus the dependency on $\delta$ is minor.

For our proof, we define the quantities
\begin{multline*}
\mathcal Q\triangleq k_1 + \ln\parentheses{\frac{4\log^*k_1}{\delta}} + \sqrt{2k_1\ln\parentheses{\frac{4\log^*k_1}{\delta}}} \\
= O\parentheses{\frac{k}{\log^* k} + \log\delta^{-1}}
= O\parentheses{\frac{k}{\log^* k}}
\end{multline*}
and
\begin{multline*}
\mathcal S\triangleq \frac{\mathcal Q+2\ln\parentheses{\frac{4\mathcal L}{\delta}}+ \sqrt{2\mathcal Q\ln\parentheses{\frac{4\mathcal L}{\delta}}}}{c\cdot e^{-c}}  \\
= O\parentheses{\frac{k}{\log^* k} + \log\delta^{-1}}
= O\parentheses{\frac{k}{\log^* k}}. 
\end{multline*}
Note that $\mathcal Q$ and $\mathcal S$ are not known to our algorithm (which is only aware of $d$) and they are used strictly for the analysis.
Our proof follows the next roadmap:
\begin{enumerate}
    \item When a flow has at least 
    $\mathcal X\triangleq k\log\log^* k\cdot (1+o(1))$ 
    packets, Baseline (layer $0$) gets at least $\mathcal X\cdot (1-o(1)) = k\log\log^* k\cdot (1+o(1))$ digests and XOR (layers $1$ and above) gets at least $\Omega(\mathcal X / \log\log^* k) = \Omega\parentheses{k}$ digests with probability $1-\delta/6$.\label{part1}
    \item When Baseline (layer $0$) gets at least $\mathcal X\cdot (1-o(1))$ digests, it decodes all hops but $k_1\triangleq \frac{k}{\log^* k}$ with probability $1-\delta/6$.\label{part2}
    \item When at least $\Omega\parentheses{k}$ packets reach XOR (layers $1$ and above), with probability $1-\delta/6$ each layer gets at least $\mathcal S$ digests.\label{part3}
    \item When a layer $\ell\in\set{1,\ldots,\mathcal L}$ gets $\mathcal S$ digests, with probability $1-\delta/6\mathcal L$, at least $\mathcal Q$ of the digests contain exactly one of the $k_\ell$ undecoded switches.\label{part4}
    \item When a layer $\ell\in\set{1,\ldots,\mathcal L-1}$ gets $\mathcal Q$ of digests that contain exactly one of the $k_\ell\triangleq k_1 / (e\duparrow\parentheses{\ell-1})$ undecoded switches, it decodes all hops but at most $k_{\ell+1}$ with probability $1-\delta/6\mathcal L$.\label{part5}
    \item When the last layer $\mathcal L$ gets $\mathcal Q$ of digests that contain exactly one of the $k_\ell$ undecoded switches, it decoded all the remaining hops with probability $1-\delta/6\mathcal L$.\label{part6}
\end{enumerate}

We then use the union bound over all bad events to conclude that the algorithm succeeds with probability at least $1-\delta$.

\subsubsection{Proof of Part~(\ref{part1})}
The first step is to observe that by a straightforward application of the Chernoff bound, since layer $0$ is chosen with probability $1/2$, the number of packets that it receives is with probability $1-\delta/6$:
$$
\mathcal X_0 = \tau\cdot \mathcal X \pm O\parentheses{\sqrt{\tau\cdot \mathcal X\cdot \log\delta^{-1}}}.
$$
Since $\mathcal X = \omega(\log \delta^{-1})$, we have that 
$$\mathcal X_0 \ge \mathcal X(\tau-o(1))= k\log\log^* k \cdot (1+o(1)).$$
\subsubsection{Proof of Part~(\ref{part2})}
Applying Lemma~\ref{lem:partialCoverage} for $\psi=\frac{1}{\log^* k}$, we get that after
\begin{multline*}
k\ln\log^* k+\log^* k\ln\delta^{-1} + \sqrt{2k\log^* k\ln\log^* k\ln\delta^{-1}}\\
= k\log\log^* k \cdot (1+o(1))
\end{multline*}
packets from layer $0$ the number of hops that are not decoded is at most  $k_1\triangleq k\cdot \psi = \frac{k}{\log^* k}$ with probability $1-\delta/6$. That is, we use $k_1$ to denote the number of undecoded hops that are left for layers $1$ and above.
\subsubsection{Proof of Part~(\ref{part3})}
When at least $\Omega(k)$ reach XOR, the number of digests that the levels get is a balls and bins processes with the levels being the bins.
According to Theorem~\ref{thm:NVal}:

After seeing 
\begin{multline*}
\mathcal L\cdot \parentheses{\mathcal S-1 + \ln(6\mathcal L/\delta) + \sqrt{(\mathcal S-1 + \ln(6\mathcal L/\delta))^2-(\mathcal S-1)^2/4}} \\=
O\parentheses{\mathcal L\cdot \parentheses{\mathcal S + \log(\mathcal L/\delta)}}\\=
O\parentheses{\log^* k\cdot \parentheses{\frac{k}{\log^* k} + \log\delta^{-1} + \log(\delta^{-1}\log^* k)}} 
=O\parentheses{k} 
\end{multline*}
packets, with probability $1-\delta/6$ our algorithm has at least $\mathcal Q$ samples in each layer.

\subsubsection{Proof of Part~(\ref{part4})}
Follows from Lemma~\ref{lem:negativeBinomialBound} for $p=c\cdot e^{-c}$, $k=\mathcal Q$ and $\delta'=\frac{\delta}{6\mathcal L}$.
\subsubsection{Proof of Part~(\ref{part5})}
Follows from Lemma~\ref{lem:partialCoverage} with $\mathcal K = k_\ell$ and $\psi=\frac{k_{\ell+1}}{k_\ell}$.
\subsubsection{Proof of Part~(\ref{part6})}
The last layer is samples needs to decode 
$$k_{\mathcal L}\le\frac{k_1}{e\duparrow\parentheses{\mathcal L-1}}=\frac{k_1}{\log \mathfrak d}=O\parentheses{\frac{k_1}{\log k_1}}$$
and samples with probability $$p_{\mathcal L} = \frac{e\duparrow (\mathcal L-1)}{\mathfrak{d}} = \frac{\log \mathfrak d}{\mathfrak d} = \Theta\parentheses{\frac{\log k_1}{k_1}}.$$
Therefore, with a constant probability, a digest would be xor-ed by exactly one of the $k_{\mathcal L}$ undecoded hops, and the number of such packets needed to decode the remainder of the path is $O\parentheses{k_{\mathcal L}\log k_{\mathcal L}} = O(k_1).$

\subsection{Revised Algorithm to Improve the Lower Order Term's Constant}
Consider changing the algorithm to sample layer $0$ with probability
$$
\tau' \triangleq \frac{1+\log\log^* d}{2+\log\log^* d} = 1 - \frac{1}{2+\log\log^* d}.
$$

Then when getting $\mathcal X' = k \cdot \parentheses{\log\log^* k + 1 + \frac{1}{ce^{1-c}} + o(1)}$, we will have 
$$
k \cdot (\log\log^* k + 1 + o(1))
$$
packets that reach layer $0$, which would leave only 
$$
k_1' \triangleq
\frac{k}{e\cdot \log^* k}
$$
undecoded hops to layers $1$ and above.
As above, the number of packets required for the upper layers to decode the missing hops is 
$$
\frac{k_1'\log^* k_1'}{ce^{-c}} \le \frac{k}{ce^{1-c}}.
$$
Since $ce^{-c}\le 1/e$ for any $c>0$, we get that this is a strict improvement in the number of packets that are required for the path decoding. 
For example, if $d = k$ (i.e., $c=1$), we reduce the required number of packets from $k(\log\log^* k + e + o(1))$ to $k(\log\log^* k + 2 + o(1))$.

\subsection{An Extension -- Detecting Routing Loops}
Real-time detection of routing loops is challenging, as switches need to recognize looping packets without storing them.
Interestingly, we can leverage \sys to detect loops on the fly. To do so, we check whether the current switch's hash matches the one on the packet. Specifically, before choosing whether to sample or not, the switch checks whether $p_j.\mbox{dig}=h(s,p_j)$. 
If there is a loop and $s$ was the last switch to write the digest, it will be detected. Unfortunately, such an approach may result in a significant number of false positives. For example, if we use $\mathfrak b=16$-bit hashes, the chance of reporting a false loop over a path of length $32$ would be roughly $0.05\%$, which means several false positives per second on a reasonable network.

To mitigate false positives, we propose requiring multiple matches, corresponding to multiple passes through the loop. We use an additional counter $c$ to track the number of times a switch hash matched the digest. When $c=0$, the switches follow the same sampling protocol as before. However, if $c>0$ then the digest is no longer changed, and if $c$ exceeds a value of $T$ then we report a loop. This changes the loop detection time, but the switch that first incremented $c$ may report the loop after at most $T$ cycles over the loop. 
This approach adds $\ceil{\log_2 {(T+1)}}$ bits of overhead, but drastically reduces the probability of false positives. For example, if $T=1$ and $\mathfrak b=15$, we still have an overhead of sixteen bits per packet, but the chance of reporting false loops decreases to less than $5\cdot 10^{-7}$. If we use $T=3$ and $\mathfrak b=14$, the false reporting rate further decreases to $5\cdot 10^{-13}$, which allows the system to operate \mbox{without false alarms in practice.}

\begin{algorithm}[!htb]
\algSize
\caption{\sys Processing at $s$ with Loop Detection}
\label{alg:loop}
\begin{algorithmic}[1]
    \Statex\textbf{Input:} A packet $p_j$ with $\mathfrak b$-bits digest $p_j.\mbox{dig}$ and a counter $p_j.c$.
    \Statex\textbf{Output:} Updated digest $p_j.\mbox{dig}$ or LOOP message.
    \If{$p_j.\mbox{dig} = h(s,p_j)$ }
        \If{$p_j.c = T$ }{ \Return LOOP}
        \EndIf    
        \State $p_j.c\gets p_j.c+1$
    \EndIf
    \Statex Let $i$ such that the current switch is the $i'th$ so far
    \If{$p_j.c=0$ and $g(p_j,i)$ < $1/i$}
            \State $p_j.\mbox{dig} \gets h(s,p_j)$\label{line:sampleOutgoing3}\Comment{Sample with probability $1/i$}
    \EndIf
\end{algorithmic}
\end{algorithm}

\section{Compting HPCC's Utilization}\label{app:hpcc}
We first calculate the logarithm:

$$U\_term=\log(\frac{T-\tau}{T}\cdot U)=\log(T-\tau)-\log T+\log U$$

$$\mathit{qlen}\_term=\log(\frac{\mathit{qlen}\cdot\tau}{B\cdot T^2})=\log \mathit{qlen}+\log\tau-\log B-2\log T$$

$$\mathit{byte}\_term=\log(\frac{\mathit{byte}}{B\cdot T})=\log \mathit{byte}-\log B-\log T$$

Then calculate $U$ using exponentiation:

$$U=2^{U\_term} + 2^{\mathit{qlen}\_term} + 2^{\mathit{byte}\_term}$$

\changed{\section{Arithmetic Operations in the Data Plane}\label{app:arithmeticImplementation}}
\changed{Some of our algorithms require operations like multiplication and division that may not be natively supported on the data plane of current programmable switches.
Nonetheless, we now discuss how to approximate these operations through fixed-point representations, logarithms, and exponentiation. We note that similar techniques have appeared, for example, in~\cite{sharma17},~\cite{Cheetah} and~\cite{ding2020estimating}.}

\changed{\parab{Fixed-point representation:}}
\changed{Modern switches may not directly support representation of fractional values. Instead, when requiring a \emph{real-valued} variable in the range $[0,R]$, we can use $m$ bits to represent it so that the integer \mbox{representation} $r\in\set{0,1,\ldots,2^{m}-1}$ stands for $R\cdot r \cdot 2^{-m}$.  $R$ is called \emph{scaling factor} and is often a power of two for simplicity.
For example, if our range is $[0,2]$, and we use $m=16$ bits, then the encoding value $39131$ represents $2\cdot 39131\cdot 2^{-16}\approx1.19$.}


\changed{Conveniently, this representation immediately allows using integer operations (e.g., addition or multiplication) to manipulate the variables. For example, if $x$ and $y$ are variables with scale factor $R$ that are represented using $r(x),r(y)$, then their sum is represented using $r(x)+r(y)$ (assuming no overflow, this keeps the scaling factor intact) and their product is $r(x)\cdot r(y)$ with a scaling factor of $R^2$. 
As a result, we hereafter consider operating on integer values.}

\changed{\parab{Computing logarithms and exponentiating:}}
\changed{
Consider needing to approximate $\log_2 (x)$ for some integer $x$ (and storing the result using a fixed-point representation).
If the domain of $x$ is small (e.g., it is an $8$-bit value), we can immediately get the value using a lookup table.
Conversely, say that $x$ is an $m$-bit value for a large $m$ (e.g., $m=64$).
In this case, we can use the switch's TCAM to find the most significant set bit in $x$, denoted $\ell$. That is, we have that $x=2^\ell\cdot \alpha$ for some $\alpha\in[1,2)$. Next, consider the next $q$ bits of $x$, denoted by $x_q$, where $q$ is such that it is feasible to store a $2^q$-sized lookup table on the switch (e.g., $q=8$).
\footnote{If $q<\ell$ we can simply look up the exact value as before. } 
Then we have that $x= x_q\cdot 2^{\ell-q} (1+\epsilon)$ for a small relative error $\epsilon<  2^{-q}$.
Therefore, we write 
$$
\log_2(x) = \log_2(x_q\cdot 2^{\ell-q} (1+\epsilon))
= (\ell-q) + \log_2 (x_q) + \log_2(1+\epsilon).
$$
Applying the lookup table to $x_q$, we can compute $\widetilde y\triangleq (\ell-q) + \log_2 (x_q)$ on the data plane and get that $\widetilde y\in [\log_2 x -  \log_2(1+\epsilon), \log_2 x]$.\footnote{In addition to the potential error that arises from the lookup table.} 
We can further simplify the error expression as $\log_2(1+\epsilon)\le \epsilon/\ln 2\approx 1.44\cdot 2^{-q}$.
We also note that computing logarithms with other bases can be done similarly as $\log_y x = \log_2 x / \log_2 y$.
}

\changed{For exponentiation, we can use a similar trick. Assume that we wish to compute $2^x$ for some real-valued $x$ that has a fixed-point representation $r$. 
Consider using a lookup table of $2^q$ entries for a suitable value of $q$, and using the TCAM to find the most significant set bit in $r$.
Then we can compute $2^x$ up to a multiplicative factor of $2^{x\epsilon}$ for some $\epsilon\le 2^{-q}$. Assuming that $x$ is bounded by $R\le 2^{q}$, this further simplifies to $2^{x\epsilon}\le 2^{x2^{-q}} \le 1+R\cdot 2^{-q}$.
For example, if $x$ is in the range $[0,2]$ and we are using $q=8$ then logarithms are computed to within a $(1+2^{-7})$-multiplicative factor \mbox{(less than 1\% error).}
}

\changed{\parab{Multiplying and dividing:}}
\changed{We overcome the lack of support for arithmetic operations such as multiplication and division using approximations, via logarithms and exponentiation.
Intuitively, we have that $x\cdot y = 2^{\log_2 x + \log_2 y}$ and $x / y = 2^{\log_2 x - \log_2 y}$. We have already discussed how to approximate logarithms and exponentiation, while addition and subtraction are currently supported. 
We note that the errors of the different approximations compound and thus it is crucial to maintain sufficient accuracy at each step to produce a meaningful approximation for the multiplication \mbox{and division operations.}}

\changed{
An alternative approach is to directly use a lookup table that takes the $q/2$ most significant bits, starting with the first set bit, of $x$ and $y$ and return their product/quotient (as before, this would require a $2^q$-sized table). However, going through logarithms may give a more accurate result as the same lookup table can be used for both $x$ and $y$, and its keys are a single value, which allows considering $q$ bits for the same memory usage.}